\documentclass[a4paper]{article}

\usepackage{Lyapunov}

\usepackage{fancyhdr}
\usepackage[round,comma,authoryear]{natbib}

\usepackage{graphicx,url}

\usepackage[english]{babel}
\usepackage[utf8]{inputenc}

\usepackage{indentfirst}
\usepackage{enumitem}

\usepackage{amsmath,amsfonts,amssymb,amscd,amsthm,xspace}
\theoremstyle{plain}
\newtheorem{theorem}{Theorem}

\newtheorem{lemma}{Lemma}
\newtheorem{proposition}{Proposition}

\theoremstyle{definition}
\newtheorem{definition}{Definition}

\newtheorem{example}{Example}

\theoremstyle{remark}
\newtheorem{remark}{Remark}

\title{Generalised Lyapunov Functions and Functionally Generated Trading Strategies\footnote{We are grateful to Andrea Macrina and Daniel Schwarz for their detailed reading and helpful comments.}}

\author{{\sc Johannes Ruf}\inst{1} ~~~ {\sc Kangjianan Xie}\inst{2}}

\address{Department of Mathematics, The London School of Economics and Political\\
  Science, Houghton Street, London WC2A 2AE, United Kingdom\\
  E-mail: \textit{j.ruf@lse.ac.uk}
\nextinstitute
Department of Mathematics, University College London\\
  Gower Street, London WC1E 6BT, United Kingdom\\
  E-mail: \textit{kangjianan.xie.14@ucl.ac.uk}
}

\begin{document}

\vspace{5mm}

\maketitle

\centerline{\today}

\begin{abstract}
This paper investigates the dependence of functional portfolio generation, introduced by \cite{F_generating}, on an extra finite variation process. The framework of \cite{MR3663643} is used to formulate conditions on trading strategies to be strong arbitrage relative to the market over sufficiently large time horizons. A mollification argument and Koml\'{o}s theorem yield a general class of potential arbitrage strategies. These theoretical results are complemented by several empirical examples using data from the S\&P 500 stocks.

\vspace{0.5cm}
\end{abstract}

\section{Introduction}

E.R.~Fernholz established Stochastic Portfolio Theory (SPT) to provide a theoretical tool for applications in equity markets, and for analyzing portfolios with controlled behavior; see \cite{F_generating} and \cite{FK_survey}, for example. SPT studies so called functionally generated portfolios. The value of a functionally generated portfolio relative to the total market capitalization is merely a function, known as the so called \textit{master formula}, of the market weights. This formula does not involve stochastic integration or drifts, which makes the analysis very easy as the need for estimation is reduced.

One very interesting topic following up this construction is the study of relative arbitrage opportunities between functionally generated portfolios and the market portfolio. \cite{F_generating, MR1861997, MR1894767} states conditions for such relative arbitrage to exist over sufficiently large time horizons. To implement this relative arbitrage, trading strategies generated by suitable portfolio generating functions are required. \cite{MR3663643} interpret portfolio generating functions as Lyapunov functions. More precisely, the supermartingale property of the corresponding wealth processes after an appropriate change of measures is utilized to study the performance of functionally generated trading strategies. Relative arbitrage over arbitrary time horizons under appropriate conditions is also studied by \cite{Volatility}.

One offspring of portfolio generating functions is a generalized portfolio generating function, which depends on an additional argument with continuous path and finite variation. This is inspired by the fact that in practice, people tend to take historical data, such as past performance of stocks, or statistical estimates, into consideration when constructing portfolios. Besides, this generalization provides additional flexibility in choosing portfolio generating functions. Section~3.2 of \cite{MR1894767} formulates the concept of time-dependent generating functions, and presents the master formula under this situation. In the same framework, \cite{MR3246899} shows an extension of the master formula to portfolios generated by functions that also depend on the current state of some continuous path process of finite variation. Also based on Fernholz's structure, \cite{Schied_2016} provide a pathwise version of the relevant master formula. They also analyze examples where the additional process is chosen to be the moving average of the market weights.

All the above mentioned papers (\cite{MR1894767}, \cite{MR3246899}, and \cite{Schied_2016}) make assumptions on the smoothness of the portfolio generating function with respect to both the finite variation process and the market weights. In this paper, we weaken these assumptions such that the choice for the portfolio generating function is less restricted. To this end, we use a mollification argument and the Koml\'{o}s theorem. Then we study several examples empirically, using data from the S\&P 500 index.\footnote{As the constituent list of the stocks in the S\&P 500 index changes over time, we avoid a survivorship bias by not restricting the analysis to the current stocks in the S\&P 500 index. Instead, we reconstruct the historical constituent list of the S\&P 500 index and adjust the portfolios appropriately when the constituent list changes.}

An outline of the paper is as follows. Section~\ref{sec 2} specifies the market model and recalls the definitions of trading strategies and relative arbitrage. Section~\ref{sec 3} first gives the definitions of regular functions and Lyapunov functions, and then presents sufficient conditions for a function to be regular and Lyapunov, respectively. The appendix presents the proofs of these results. Section~\ref{sec 4} defines additive and multiplicative generation, and the corresponding trading strategies and wealth processes. Section~ \ref{sec 4} also gives conditions for arbitrage relative to the market portfolio to exist. Section~\ref{sec 5} describes the data involved and the processing method to implement the empirical analysis. Section~\ref{sec 6} contains several examples of portfolio generating functions and discusses empirical results. Section~\ref{sec 7} concludes.

\vspace{0.5cm}

\section{Model setup}\label{sec 2}

We fix a filtered probability space $\big(\Omega,\mathcal{F}(\infty),\mathcal{F}(\cdot),\textsf{P}\big)$ with $\mathcal{F}(0)=\{\emptyset,\Omega\}$ and write
\[
\Delta^{d}=\left\{(x_{1},\cdots,x_{d})'\in[0,1]^{d}:\sum_{i=1}^{d}x_{i}=1\right\}\quad\text{and}\quad\Delta_{+}^{d}=\Delta^{d}\cap(0,1)^{d}.
\]

We consider an equity market with $d\geq2$ companies. We denote the market weights process by $\mu(\cdot)=\big(\mu_{1}(\cdot),\cdots,\mu_{d}(\cdot)\big)'$. Here, $\mu_{i}(\cdot)$ is the market weight process of company $i$ computed by dividing the capitalization of company $i$ by the total capitalization of all $d$ companies in the market, for all $i\in\{1,\cdots,d\}$. We assume that $\mu(\cdot)$ is $\Delta^{d}$-valued with $\mu(0)\in\Delta_{+}^{d}$, and that $\mu_{i}(\cdot)$ is a continuous, non-negative semimartingale, for all $i\in\{1,\cdots,d\}$.

To define a trading strategy for $\mu(\cdot)$, let us consider a process $\vartheta(\cdot)=\big(\vartheta_{1}(\cdot),\cdots,\vartheta_{d}(\cdot)\big)'$ in $\mathbb{R}^{d}$, which is predictable and integrable with respect to $\mu(\cdot)$. We denote the collection of all such processes by $\mathcal{L}(\mu)$.

For such a process $\vartheta(\cdot)\in\mathcal{L}(\mu)$, we interpret $\vartheta_{i}(t)$ as the number of shares in the stock of company $i$ held at time $t\geq0$, for all $i\in\{1,\cdots,d\}$. Then
\[
V^{\vartheta}(\cdot)=\sum_{i=1}^{d}\vartheta_{i}(\cdot)\mu_{i}(\cdot)
\]
can be interpreted as the wealth process corresponding to $\vartheta(\cdot)$.
\\
\begin{definition}\label{Definition 2.2} (Trading strategies). A process $\varphi(\cdot)\in\mathcal{L}(\mu)$ is called a \textit{trading strategy} if
\[
V^{\varphi}(\cdot)-V^{\varphi}(0)=\int_{0}^{\cdot}\sum_{i=1}^{d}\varphi_{i}(t)\mathrm{d}\mu_{i}(t).
\]
\qed
\end{definition}
\vspace{0.5cm}
\begin{remark}\label{Remark nonself}
To convert a predictable process $\vartheta(\cdot)\in\mathcal{L}(\mu)$ into a trading strategy $\varphi(\cdot)$, we adapt the measure of the ``defect of self-financibility'' of $\vartheta(\cdot)$, introduced in Proposition~2.4 in \cite{MR3663643} and defined as
\begin{equation}\label{eq 2.7}
Q^{\vartheta}(\cdot)=V^{\vartheta}(\cdot)-V^{\vartheta}(0)-\int_{0}^{\cdot}\sum_{i=1}^{d}\vartheta_{i}(t)\mathrm{d}\mu_{i}(t).
\end{equation}
As a result, the process $\varphi(\cdot)$ with components
\begin{equation}\label{eq 2.8}
\varphi_{i}(\cdot)=\vartheta_{i}(\cdot)-Q^{\vartheta}(\cdot)+C,\quad i\in\{1,\cdots,d\},
\end{equation}
where $C$ can be any real constant, is a trading strategy for $\mu(\cdot)$.
\qed
\end{remark}

We are interested in the performances of different portfolios. Especially, we focus on studying the conditions for the existence of so called relative arbitrage.
\\
\begin{definition}\label{Definition realtive_arbitrage} (Arbitrage relative to the market). A trading strategy $\varphi(\cdot)$ is said to be \textit{relative arbitrage} with respect to the market over a given time horizon $[0,T]$, for $T\geq0$, if
\[
V^{\varphi}(\cdot)\geq0\quad\text{and}\quad V^{\varphi}(0)=1,
\]
along with
\begin{equation}\label{eq 2.13}
\textsf{P}\big[V^{\varphi}(T)\geq1\big]=1\quad\text{and}\quad\textsf{P}\big[V^{\varphi}(T)>1\big]>0.
\end{equation}

If $\textsf{P}\big[V^{\varphi}(T)>1\big]=1$ holds, we say that the relative arbitrage is strong over $[0,T]$.
\qed
\end{definition}
\vspace{0.5cm}
\begin{remark}\label{Remark 2.5}
Definition~\ref{Definition realtive_arbitrage} makes sense due to the fact that the wealth process of the market portfolio at any time is given by
\[
V^{(1,\cdots,1)}(\cdot)=\sum_{i=1}^{d}\mu_{i}(\cdot)=1.
\]
Then relative arbitrage exists over a given time horizon $[0,T]$ when a non-negative wealth process $V^{\varphi}(\cdot)$ has the same initial wealth as the market portfolio, the probability for $V^{\varphi}(T)$ to be greater than the market portfolio is strictly positive, and $V^{\varphi}(T)$ is not lower than the market portfolio.
\qed
\end{remark}

In the following sections, we study portfolio generating functions that depend on some $\mathbb{R}^{m}$-valued continuous process of finite variation on $[0,T]$, for $T\geq0$ and some $m\in\mathbb{N}$. We use $\Lambda(\cdot)$ to denote such a process. This process allows for more flexibility in selecting portfolio generating functions. To this end, let $\mathcal{W}$ be some open subset of $\mathbb{R}^{m}\times\mathbb{R}^{d}$ such that
\begin{equation}\label{eq W}
\textsf{P}\big[\big(\Lambda(t),\mu(t)\big)\in\mathcal{W},~\forall~t\geq0\big]=1.
\end{equation}

The following notations are introduced for the ranked market weights, which are studied in Theorem~\ref{Theorem new2} and Example~\ref{Example Clip}. For a vector $x=(x_{1},\cdots,x_{d})'\in\Delta^{d}$, denote its corresponding ranked vector as $\boldsymbol{x}=\big(x_{(1)},\cdots,x_{(d)}\big)'$, where
\[
\max_{i\in\{1,\cdots,d\}}x_{i}=x_{(1)}\geq x_{(2)}\geq\cdots\geq x_{(d-1)}\geq x_{(d)}=\min_{i\in\{1,\cdots,d\}}x_{i}
\]
are the components of $x$ in descending order. Denote
\[
\mathbb{W}^{d}=\left\{\left(x_{(1)},\cdots,x_{(d)}\right)'\in\Delta^{d}:1\geq x_{(1)}\geq x_{(2)}\geq\cdots\geq x_{(d-1)}\geq x_{(d)}\geq0\right\};
\]
then the rank operator $\mathfrak{R}:\Delta^{d}\rightarrow\mathbb{W}^{d}$ is a mapping such that $\mathfrak{R}(x)=\boldsymbol{x}$. Moreover, denote $\mathbb{W}_{+}^{d}=\mathbb{W}^{d}\cap(0,1)^{d}$.

The ranked market weights process $\boldsymbol{\mu}(\cdot)$ is given by
\begin{equation}\label{eq rank mu}
\boldsymbol{\mu}(\cdot)=\mathfrak{R}\big(\mu(\cdot)\big)=\big(\mu_{(1)}(\cdot),\cdots,\mu_{(d)}(\cdot)\big)',
\end{equation}
which is itself a continuous, $\Delta^{d}$-valued semimartingale whenever $\mu(\cdot)$ is a continuous, $\Delta^{d}$-valued semimartingale (see Theorem~2.2 in \cite{MR2428716}). At last, let $\boldsymbol{W}$ be some open subset of $\mathbb{R}^{m}\times\mathbb{R}^{d}$ such that
\begin{equation}\label{eq W1}
\textsf{P}\big[\big(\Lambda(t),\boldsymbol{\mu}(t)\big)\in\boldsymbol{W},~\forall~t\geq0\big]=1.
\end{equation}

\vspace{0.5cm}

\section{Generalized regular and Lyapunov functions}\label{sec 3}

In this section, we consider two classes of portfolio generating functions, regular and Lyapunov functions, which are introduced in \cite{MR3663643}. We generalize these notions here to allow for the additional process $\Lambda(\cdot)$.
\\
\begin{definition}\label{Definition 3.1}
(Regular function). A continuous function $G:\mathcal{W}\rightarrow\mathbb{R}$ is said to be \textit{regular} for $\Lambda(\cdot)$ and $\mu(\cdot)$ if
\begin{enumerate}
  \item there exists a measurable function $DG=(D_{1}G,\cdots,D_{d}G)':\mathcal{W}\rightarrow\mathbb{R}^{d}$ such that the process $\boldsymbol{\vartheta}(\cdot)=\big(\boldsymbol{\vartheta}_{1}(\cdot),\cdots,\boldsymbol{\vartheta}_{d}(\cdot)\big)'$ with components
      \begin{equation}\label{eq 3.1}
      \boldsymbol{\vartheta}_{i}(\cdot)=D_{i}G\big(\Lambda(\cdot),\mu(\cdot)\big),\quad i\in\{1,\cdots,d\},
      \end{equation}
      is in $\mathcal{L}(\mu)$; and
  \item the continuous, adapted process
  \begin{equation}\label{eq 3.2}
  \Gamma^{G}(\cdot)=G\big(\Lambda(0),\mu(0)\big)-G\big(\Lambda(\cdot),\mu(\cdot)\big)+\int_{0}^{\cdot}\sum_{i=1}^{d}\vartheta_{i}(t)\mathrm{d}\mu_{i}(t)
  \end{equation}
  is of finite variation on the interval $[0,T]$, for all $T\geq0$.
\end{enumerate}
\qed
\end{definition}
\vspace{0.5cm}
\begin{definition}\label{Definition Generalized Lyapunov functions}
(Lyapunov function). A regular function $G:\mathcal{W}\rightarrow\mathbb{R}$ is said to be a {\textit{Lyapunov function}} for $\Lambda(\cdot)$ and $\mu(\cdot)$ if, for some function $DG$ as in Definition~\ref{Definition 3.1}, the finite variation process $\Gamma^{G}(\cdot)$ of \eqref{eq 3.2} is non-decreasing.
\qed
\end{definition}

In the next example, we discuss sufficient conditions for a smooth function to be regular or Lyapunov.
\\
\begin{example}\label{Example 3.5}
Consider a $C^{1,2}$ function $G:\mathcal{W}\rightarrow\mathbb{R}$. Then It\^{o}'s lemma yields
\[
\begin{aligned}
G\big(\Lambda(\cdot),\mu(\cdot)\big)&=G\big(\Lambda(0),\mu(0)\big)+\int_{0}^{\cdot}\sum_{v=1}^{m}\frac{\partial G}{\partial \lambda_{v}}\big(\Lambda(t),\mu(t)\big)\mathrm{d}\Lambda_{v}(t)&\\
&+\int_{0}^{\cdot}\sum_{i=1}^{d}\frac{\partial G}{\partial x_{i}}\big(\Lambda(t),\mu(t)\big)\mathrm{d}\mu_{i}(t)&\\
&+\frac{1}{2}\sum_{i,j=1}^{d}\int_{0}^{\cdot}\frac{\partial^{2}G}{\partial x_{i}\partial x_{j}}\big(\Lambda(t),\mu(t)\big)\mathrm{d}\big\langle\mu_{i},\mu_{j}\big\rangle(t).&
\end{aligned}
\]

Set now $\boldsymbol{\vartheta}_{i}(\cdot)=\frac{\partial G(\Lambda(\cdot),\mu(\cdot))}{\partial x_{i}}$ and
\begin{equation}\label{eq 3.4}
\begin{aligned}
\Gamma^{G}(\cdot)=&-\int_{0}^{\cdot}\sum_{v=1}^{m}\frac{\partial G}{\partial \lambda_{v}}\big(\Lambda(t),\mu(t)\big)\mathrm{d}\Lambda_{v}(t)
&\\
&-\frac{1}{2}\sum_{i,j=1}^{d}\int_{0}^{\cdot}\frac{\partial^{2}G}{\partial x_{i}\partial x_{j}}\big(\Lambda(t),\mu(t)\big)\mathrm{d}\big\langle\mu_{i},\mu_{j}\big\rangle(t).&
\end{aligned}
\end{equation}
Obviously, the process $\Gamma^{G}(\cdot)$ has finite variation on $[0,T]$, for all $T\geq0$. Hence $G$ is a regular function.

Moreover, if the process $\Gamma^{G}(\cdot)$ is non-decreasing, then $G$ is not only a regular function, but also a Lyapunov function. For instance, this holds if $G$ is non-decreasing in every dimension with respect to the first argument and $\Lambda(\cdot)$ is decreasing in every dimension, and $G$ is concave with respect to the second argument.
\qed
\end{example}

Below we give sufficient conditions for a function $G$ to be regular (Lyapunov). To this end, recall the open set $\mathcal{W}$ from \eqref{eq W}.
\\
\begin{theorem}\label{Theorem new1}
For a continuous function $G:\mathcal{W}\rightarrow\mathbb{R}$, consider the following conditions.
\begin{enumerate}
\item[(ai)] On any compact set $\mathcal{\overline{V}}\subset\mathcal{W}$, there exists a constant $L=L(\mathcal{\overline{V}})\geq0$ such that, for all $(\lambda_{1},x),(\lambda_{2},x)\in\mathcal{\overline{V}}$,
    \[
    \left|G(\lambda_{1},x)-G(\lambda_{2},x)\right|\leq L\|\lambda_{1}-\lambda_{2}\|_{2}.
    \]
\item[(aii)] $G(\cdot,x)$ is non-increasing, for fixed $x$, and $\Lambda(\cdot)$ is non-decreasing in every dimension.
\item[(bi)] $G$ is differentiable in the second argument. Moreover, on any compact set $\mathcal{\overline{V}}\subset\mathcal{W}$, there exists a constant $L=L(\mathcal{\overline{V}})\geq0$ such that, for all $(\lambda,x_{1}),(\lambda,x_{2})\in\mathcal{\overline{V}}$,
    \[
    \left\|\frac{\partial G}{\partial x}(\lambda,x_{1})-\frac{\partial G}{\partial x}(\lambda,x_{2})\right\|_{2}\leq L\|x_{1}-x_{2}\|_{2}.
    \]
\item[(bii)] $G(\lambda,\cdot)$ is concave, for fixed $\lambda$.
\end{enumerate}
If one of the conditions (ai) or (aii) holds and one of the conditions (bi) or (bii) holds, $G$ is a regular function for $\Lambda(\cdot)$ and $\mu(\cdot)$. Moreover, in the case that (aii) and (bii) hold, $G$ is Lyapunov.
\end{theorem}

\vspace{0.5cm}

The proof of Theorem~\ref{Theorem new1} is given in the appendix. A generalized version of It\^{o}'s formula studied in \cite{MR2723141} is related but can only be applied in a Markovian setting.

\vspace{0.5cm}

\begin{remark}\label{Remark 3.6}
In contrast to Theorem~$3.7$ in \cite{MR3663643}, even if $G$ can be extended to a continuous function concave in the second argument, $G$ may not be Lyapunov. A counterexample is given in Example~\ref{Example 3.7}. Therefore, for the generalized case, Theorem~$3.7$ in \cite{MR3663643} cannot be applied, and instead we have to use modified conditions such as given by Theorem~\ref{Theorem new1}.
\qed
\end{remark}

\vspace{0.5cm}

\begin{example}\label{Example 3.7}
Assume that $\mu(\cdot)\in\Delta^{d}$ with $\langle\mu_{1},\mu_{1}\rangle(t)>0$, for all $t>0$, and that
\[
\Lambda(\cdot)=\gamma\sum_{i=1}^{d}\langle\mu_{i},\mu_{i}\rangle(\cdot),
\]
where $\gamma$ is a constant.

Define the concave quadratic function
\[
G(\lambda,x)=\lambda-\sum_{i=1}^{d}x^{2}_{i},\quad\lambda\in\mathbb{R},~x\in\Delta^{d}.
\]
Then from \eqref{eq 3.4} we have
\[
\Gamma^{G}(\cdot)=-\int_{0}^{\cdot}\mathrm{d}\Lambda(t)+\sum_{i=1}^{d}\int_{0}^{\cdot}\mathrm{d}\langle\mu_{i},\mu_{i}\rangle(t)
=(1-\gamma)\sum_{i=1}^{d}\langle\mu_{i},\mu_{i}\rangle(\cdot).
\]
Observe that $\Gamma^{G}(\cdot)$ is decreasing for $\gamma>1$; hence $G$ is not a Lyapunov function for $\Lambda(\cdot)$ and $\mu(\cdot)$, although it is concave in its second argument.

Define now $\overline{G}(\lambda,x)=-G(\lambda,x)$. Then we have $\Gamma^{\overline{G}}(\cdot)=-\Gamma^{G}(\cdot)$. Therefore, if $\gamma>1$ holds, $\Gamma^{\overline{G}}(\cdot)$ is increasing; hence $\overline{G}$ is Lyapunov although convex in its second argument.
\qed
\end{example}

Recall the ranked market weights process $\boldsymbol{\mu}(\cdot)$ defined in \eqref{eq rank mu} and the open set $\boldsymbol{W}$ from \eqref{eq W1}.
\\
\begin{theorem}\label{Theorem new2}
If a function $\boldsymbol{G}:\boldsymbol{W}\rightarrow\mathbb{R}$ is regular for $\Lambda(\cdot)$ and $\boldsymbol{\mu}(\cdot)=\mathfrak{R}\big(\mu(\cdot)\big)$, then the composition $G=\boldsymbol{G}\circ\mathfrak{R}$ is regular for $\Lambda(\cdot)$ and $\mu(\cdot)$.
\end{theorem}

\vspace{5mm}

We refer to the appendix for the proof of Theorem~\ref{Theorem new2}.

The following example concerns a function $G$ which is not in $C^{1,2}$.
\\
\begin{example}\label{Example Clip}
Assume that $\mu(\cdot)\in\Delta_{+}^{d}$ and consider the $C^{1,2}$ function
\[
\boldsymbol{G}(\lambda,\boldsymbol{x})=-\lambda\sum_{l=1}^{d_{1}}x_{(l)}\log x_{(l)}+1-\sum_{l=d_{1}+1}^{d_{2}}x_{(l)}^{2},\quad\lambda\in\mathbb{R},~\boldsymbol{x}\in\mathbb{W}_{+}^{d},
\]
where $d_{1}$ and $d_{2}$ are positive integers with $d_{1}<d_{2}\leq d$. According to Example~\ref{Example 3.5}, $\boldsymbol{G}$ is regular for $\Lambda(\cdot)$ and $\boldsymbol{\mu}(\cdot)$. In particular, the corresponding measurable function $D\boldsymbol{G}$ as in Definition~\ref{Definition 3.1} can be chosen with components
\begin{equation}\label{eq DG1}
\begin{aligned}
D_{l}\boldsymbol{G}(\lambda,\boldsymbol{x})=
\begin{cases}
-\lambda\log x_{(l)}-\lambda,\quad&\text{if}~l\in\{1,\cdots,d_{1}\}\\
-2x_{(l)},\quad&\text{if}~l\in\{d_{1}+1,\cdots,d_{2}\}\\
0,\quad&\text{otherwise}
\end{cases}
\end{aligned}.
\end{equation}
In this case, It\^{o}'s lemma yields
\begin{equation}\label{eq G2}
\begin{aligned}
\boldsymbol{G}\big(\Lambda(\cdot),\boldsymbol{\mu}(\cdot)\big)=&\boldsymbol{G}\big(\Lambda(0),\boldsymbol{\mu}(0)\big)+
\int_{0}^{\cdot}\sum_{l=1}^{d}D_{l}\boldsymbol{G}\big(\Lambda(t),\boldsymbol{\mu}(t)\big)\mathrm{d}\mu_{(l)}(t)-\Gamma^{\boldsymbol{G}}(\cdot)\\
&+\int_{0}^{\cdot}\sum_{l=1}^{d_{1}}\mu_{(l)}(t)\log\mu_{(l)}(t)\mathrm{d}\Lambda(t)
\end{aligned}
\end{equation}
with $D_{l}\boldsymbol{G}$ given in \eqref{eq DG1} and
\begin{equation}\label{eq Gamma2}
\Gamma^{\boldsymbol{G}}(\cdot)=\frac{1}{2}\int_{0}^{\cdot}\sum_{l=1}^{d_{1}}\frac{\Lambda(t)}{\mu_{(l)}(t)}\mathrm{d}\left\langle\mu_{(l)},\mu_{(l)}\right\rangle(t)
+\int_{0}^{\cdot}\sum_{l=d_{1}+1}^{d_{2}}\mathrm{d}\left\langle\mu_{(l)},\mu_{(l)}\right\rangle(t).
\end{equation}

Denote the number of components of $x=(x_{1},\cdots,x_{d})'\in\Delta^{d}$ that coalesce at a given rank $l\in\{1,\cdots,d\}$ by
\[
N_{l}(x)=\sum_{i=1}^{d}\boldsymbol{1}_{\left\{x_{i}=x_{(l)}\right\}}.
\]
Then by Theorem~2.3 in \cite{MR2428716}, the ranked market weights process $\boldsymbol{\mu}(\cdot)$ has components
\begin{equation}\label{eq mu1}
\begin{aligned}
\mu_{(l)}(\cdot)=&\mu_{(l)}(0)+\int_{0}^{\cdot}\sum_{i=1}^{d}\frac{\boldsymbol{1}_{\left\{\mu_{i}(t)=\mu_{(l)}(t)\right\}}}{N_{l}\big(\mu(t)\big)}
\mathrm{d}\mu_{i}(t)+\sum_{k=l+1}^{d}\int_{0}^{\cdot}\frac{\mathrm{d}\boldsymbol{\Lambda}^{(l,k)}(t)}{N_{l}\big(\mu(t)\big)}\\
&-\sum_{k=1}^{l-1}\int_{0}^{\cdot}\frac{\mathrm{d}\boldsymbol{\Lambda}^{(k,l)}(t)}{N_{l}\big(\mu(t)\big)},\quad l\in\{1,\cdots,d\},
\end{aligned}
\end{equation}
where $\boldsymbol{\Lambda}^{(i,j)}(\cdot)$ with $1\leq i<j\leq d$ is the local time process of the continuous semimartingale $\mu_{(i)}(\cdot)-\mu_{(j)}(\cdot)\geq0$ at the origin.

By Theorem~\ref{Theorem new2}, the function
\[
G(\lambda,x)=\boldsymbol{G}\big(\lambda,\mathfrak{R}(x)\big)=-\lambda\sum_{l=1}^{d_{1}}\sum_{i=1}^{d}\frac{\boldsymbol{1}_{\left\{x_{i}=x_{(l)}\right\}}}{N_{l}(x)}x_{i}\log x_{i}+1-\sum_{l=d_{1}+1}^{d_{2}}\sum_{i=1}^{d}\frac{\boldsymbol{1}_{\left\{x_{i}=x_{(l)}\right\}}}{N_{l}(x)}x_{i}^{2}
\]
is regular for $\Lambda(\cdot)$ and $\mu(\cdot)$, since $\boldsymbol{G}$ is regular for $\Lambda(\cdot)$ and $\boldsymbol{\mu}(\cdot)$.

Now, let us assume that $\Lambda(\cdot)$ is of the form
\[
\Lambda(\cdot)=\overline{\xi}\wedge\left(\underline{\xi}\vee\Lambda'(\cdot)\right),
\]
where $\overline{\xi}$ and $\underline{\xi}$ are two positive constants with $\underline{\xi}<\overline{\xi}$, and the process $\Lambda'(\cdot)$ is of finite variation. Let $\mathcal{G}(\lambda',x)=G\left(\overline{\xi}\wedge\left(\underline{\xi}\vee\lambda'\right),x\right)$, for all $\lambda'\in\mathbb{R}$ and $x\in\Delta^{d}_{+}$. Then with $D_{l}\boldsymbol{G}$ and $\Gamma^{\boldsymbol{G}}(\cdot)$ given in \eqref{eq DG1} and \eqref{eq Gamma2}, respectively, inserting \eqref{eq mu1} into \eqref{eq G2} yields
\[
\mathcal{G}\big(\Lambda'(\cdot),\mu(\cdot)\big)=\mathcal{G}\big(\Lambda'(0),\mu(0)\big)+
\int_{0}^{\cdot}\sum_{i=1}^{d}D_{i}\mathcal{G}\big(\Lambda'(t),\mu(t)\big)\mathrm{d}\mu_{i}(t)-\Gamma^{\mathcal{G}}(\cdot),
\]
where
\[
D_{i}\mathcal{G}(\lambda',x)=\sum_{l=1}^{d}\frac{\boldsymbol{1}_{\left\{x_{i}= x_{(l)}\right\}}}{N_{l}(x)}
D_{l}\boldsymbol{G}\big(\overline{\xi}\wedge\left(\underline{\xi}\vee\lambda'\right),\mathfrak{R}(x)\big),\quad i\in\{1,\cdots,d\},
\]
and
\[
\begin{aligned}
\Gamma^{\mathcal{G}}(\cdot)=&\Gamma^{\boldsymbol{G}}(\cdot)-\int_{0}^{\cdot}\sum_{l=1}^{d_{1}}\mu_{(l)}(t)\log\mu_{(l)}(t)
\boldsymbol{1}_{\left\{\underline{\xi}\leq\Lambda'(t)\leq\overline{\xi}\right\}}\mathrm{d}\Lambda'(t)\\
&-
\sum_{l=1}^{d-1}\sum_{k=l+1}^{d}\int_{0}^{\cdot}\frac{D_{l}\boldsymbol{G}\big(\Lambda(t),\mathfrak{R}(\mu(t))\big)}{N_{l}\big(\mu(t)\big)}\mathrm{d}\boldsymbol{\Lambda}^{(l,k)}(t)\\
&+\sum_{l=2}^{d}\sum_{k=1}^{l-1}\int_{0}^{\cdot}\frac{D_{l}\boldsymbol{G}\big(\Lambda(t),\mathfrak{R}(\mu(t))\big)}{N_{l}\big(\mu(t)\big)}\mathrm{d}\boldsymbol{\Lambda}^{(k,l)}(t)
.
\end{aligned}
\]
Observe that $\mathcal{G}$ is regular for $\Lambda'(\cdot)$ and $\mu(\cdot)$, yet it is not in $C^{1,2}$.
\qed
\end{example}

\vspace{0.5cm}

\section{Functional generation and relative arbitrage}\label{sec 4}

In \cite{MR3663643}, two types of functional generation, additive and multiplicative generation, are constructed to study the properties of relative values of functionally generated portfolios. In this section, we first discuss the generalized versions of these functional generations and corresponding properties. Then we consider sufficient conditions for strong arbitrage relative to the market to exist.

\subsection{Additive generation}

Recall the open set $\mathcal{W}$ from \eqref{eq W}.
\vspace{5mm}
\begin{definition}\label{Definition 4.1}
(Additive generation). For a regular function $G:\mathcal{W}\rightarrow\mathbb{R}$ and the process $\boldsymbol{\vartheta}(\cdot)$ given in \eqref{eq 3.1}, the trading strategy $\boldsymbol{\varphi}(\cdot)$ with components
\[
\boldsymbol{\varphi}_{i}(\cdot)=\boldsymbol{\vartheta}_{i}(\cdot)-Q^{\boldsymbol{\vartheta}}(\cdot)+C,\quad i\in\{1,\cdots,d\},
\]
in the manner of \eqref{eq 2.8} and \eqref{eq 2.7}, and with the real constant
\begin{equation}\label{eq 4.2}
C=G\big(\Lambda(0),\mu(0)\big)-\sum_{j=1}^{d}\mu_{j}(0)D_{j}G\big(\Lambda(0),\mu(0)\big),
\end{equation}
is said to be \textit{additively generated} by the regular function $G$.
\qed
\end{definition}

\vspace{0.5cm}

\begin{proposition}\label{Proposition 4.3}
The trading strategy $\boldsymbol{\varphi}(\cdot)$, generated additively by a regular function $G:\mathcal{W}\rightarrow\mathbb{R}$, has components
\begin{equation}\label{eq 4.4}
\begin{aligned}
\boldsymbol{\varphi}_{i}(\cdot)=&D_{i}G\big(\Lambda(\cdot),\mu(\cdot)\big)+\Gamma^{G}(\cdot)+G\big(\Lambda(\cdot),\mu(\cdot)\big)
-\sum_{j=1}^{d}\mu_{j}(\cdot)D_{j}G\big(\Lambda(\cdot),\mu(\cdot)\big),
\end{aligned}
\end{equation}
for all $i\in\{1,\cdots,d\}$. Moreover, the wealth process of $\boldsymbol{\varphi}(\cdot)$ is given by
\begin{equation}\label{eq 4.3}
V^{\boldsymbol{\varphi}}(\cdot)=G\big(\Lambda(\cdot),\mu(\cdot)\big)+\Gamma^{G}(\cdot).
\end{equation}
\end{proposition}

\begin{proof}
The proposition is proven exactly like Proposition~$4.3$ in \cite{MR3663643}.
\end{proof}

\subsection{Multiplicative generation}

\begin{definition}\label{Definition 4.6}
(Multiplicative generation). For a regular function $G:\mathcal{W}\rightarrow(0,\infty)$, let the process $\boldsymbol{\vartheta}(\cdot)$ be given in \eqref{eq 3.1} and assume that $1/G\big(\Lambda(\cdot),\mu(\cdot)\big)$ is locally bounded. Consider the process $\widetilde{\boldsymbol{\vartheta}}(\cdot)\in\mathcal{L}(\mu)$ with components
\begin{equation}\label{eq theta2}
\widetilde{\boldsymbol{\vartheta}}_{i}(\cdot)=\boldsymbol{\vartheta}_{i}(\cdot)\times\exp\left(\int_{0}^{\cdot}\frac{\mathrm{d}\Gamma^{G}(t)}{G\big(\Lambda(t),\mu(t)\big)}\right)
,\quad i\in\{1,\cdots,d\}.
\end{equation}
Then the trading strategy $\boldsymbol{\psi}(\cdot)$ with components
\[
\boldsymbol{\psi}_{i}(\cdot)=\widetilde{\boldsymbol{\vartheta}}_{i}(\cdot)-Q^{\widetilde{\boldsymbol{\vartheta}}}(\cdot)+C,\quad i\in\{1,\cdots,d\},
\]
in the manner of \eqref{eq 2.8} and \eqref{eq 2.7}, and with $C$ given in \eqref{eq 4.2}, is said to be \textit{multiplicatively generated} by the regular function $G$.
\qed
\end{definition}

\vspace{0.5cm}

\begin{proposition}\label{Proposition 4.7}
The trading strategy $\boldsymbol{\psi}(\cdot)$, generated multiplicatively by a regular function $G:\mathcal{W}\rightarrow(0,\infty)$ with $1/G\big(\Lambda(\cdot),\mu(\cdot)\big)$ locally bounded, has components
\begin{equation}\label{eq 4.10}
\boldsymbol{\psi}_{i}(\cdot)=V^{\boldsymbol{\psi}}(\cdot)\left(1+\frac{D_{i}G\big(\Lambda(\cdot),\mu(\cdot)\big)-
\sum_{j=1}^{d}\mu_{j}(\cdot)D_{j}G\big(\Lambda(\cdot),\mu(\cdot)\big)}{G\big(\Lambda(\cdot),\mu(\cdot)\big)}\right),
\end{equation}
for all $i\in\{1,\cdots,d\}$, where the wealth process of $\boldsymbol{\psi}(\cdot)$ is given by
\begin{equation}\label{eq 4.9}
V^{\boldsymbol{\psi}}(\cdot)=G\big(\Lambda(\cdot),\mu(\cdot)\big)\exp\left(\int_{0}^{\cdot}\frac{\mathrm{d}\Gamma^{G}(t)}{G\big(\Lambda(t),\mu(t)\big)}\right)>0.
\end{equation}
\end{proposition}

\begin{proof}
The same argument as in Proposition~$4.7$ in \cite{MR3663643} applies.
\end{proof}

\subsection{Sufficient conditions for arbitrage relative to the market}

In \cite{MR3663643}, Theorem~5.1 and Theorem~5.2 give sufficient conditions for strong arbitrage relative to the market to exist for both additively and multiplicatively generated portfolios, respectively. These results still hold for a regular / Lyapunov function $G:\mathcal{W}\rightarrow[0,\infty)$ under specific conditions.

To be consistent with the conditions of arbitrage relative to the market in \eqref{eq 2.13}, we normalize $G\big(\Lambda(0),\mu(0)\big)=1$ such that both of the wealth processes in \eqref{eq 4.3} and \eqref{eq 4.9} have initial values $1$. This normalization is guaranteed by replacing $G$ with $G+1$ when $G\big(\Lambda(0),\mu(0)\big)=0$, or with $G/G\big(\Lambda(0),\mu(0)\big)$ when $G\big(\Lambda(0),\mu(0)\big)>0$.
\\
\begin{theorem}\label{Theorem 5.1}
Fix a Lyapunov function $G:\mathcal{W}\rightarrow[0,\infty)$ with $G\big(\Lambda(0),\mu(0)\big)=1$. For some real number $T_{\ast}>0$, suppose that
\[
\mathrm{\textsf{P}}\big[\Gamma^{G}(T_{\ast})>1\big]=1.
\]
Then the additively generated trading strategy $\boldsymbol{\varphi}(\cdot)$ of Definition~\ref{Definition 4.1} is strong arbitrage relative to the market over every time horizon $[0,T]$ with $T\geq T_{\ast}$.
\end{theorem}

\begin{proof}
Use the same reasoning as in the proof of Theorem~5.1 in \cite{MR3663643}.
\end{proof}

\begin{theorem}\label{Theorem 5.2}
Assume that $|\Lambda(\cdot)|$ is uniformly bounded. Fix a regular function $G:\mathcal{W}\rightarrow[0,\infty)$ with $G\big(\Lambda(0),\mu(0)\big)=1$. For some real numbers $T_{\ast}>0$, suppose that we can find an $\varepsilon=\varepsilon(T_{\ast})>0$ such that
\[
\textsf{P}\big[\Gamma^{G}(T_{\ast})>1+\varepsilon\big]=1.
\]
Then there exists a constant $c=c(T_{\ast},\varepsilon)>0$ such that the trading strategy $\boldsymbol{\psi}^{(c)}(\cdot)$, generated multiplicatively by the regular function
\[
G^{(c)}=\frac{G+c}{1+c}
\]
as in Definition~\ref{Definition 4.6}, is strong arbitrage relative to the market over the time horizon $[0,T_{\ast}]$. Moreover, if $G$ is a Lyapunov function, then $\boldsymbol{\psi}^{(c)}(\cdot)$ is also a strong relative arbitrage over every time horizon $[0,T]$ with $T\geq T_{\ast}$.
\end{theorem}

\begin{proof}
See the proof of Theorem~5.2 in \cite{MR3663643}. Note that $G\big(\Lambda(\cdot),\mu(\cdot)\big)$ is uniformly bounded thanks to the assumptions.
\end{proof}

\vspace{0.5cm}

\section{Data source and processing}\label{sec 5}

We start this section by describing the data used in the next section, where several trading strategies are implemented. Then we discuss the method to process the data.

\subsection{Data source and description}

We shall consider a market consisting of all stocks in the S\&P 500 index. We are interested in the beginning of day and the end of day market weights of each of these stocks. To calculate these market weights accurately (according to the method in Subsection~\ref{Subsec 5.2}), we make use of two time series: the daily market values (market capitalizations, which exclude all the dividend payments) and the daily return indexes (used to consider the effect of reinvestment of dividend payments) of the corresponding component stocks in the S\&P 500 index. Both of these time series are available at the end of each trading day.

The data of the market values and return indexes is downloaded from DataStream\footnote{DataStream, operated by Thomson Reuters, is a financial time series database; see https://financial.thomsonreuters.com/en/ products/tools-applications/trading-investment-tools/datastream-macroeconomic-analysis.html.}. The first day, for which the data is available on DataStream, is September 29$^{\mathrm{th}}$, 1989. Since then there are in total 1140 constituents that have belonged to the S\&P 500 index. A list of stocks in the S\&P 500 index is also attainable on DataStream. In particular, for each month, we derive the list of constituents of the index at the last day of this month. For a constituent delisted from the index in that month, we keep it in our portfolio provided that the constituent still remains in the market till the end of that month. However, we get rid of it from our portfolio on the same day when the constituent does no longer exist in the market, usually due to mergers and acquisitions, bankruptcies, etc. For a constituent newly added to the index in that month, we put it into our portfolio from the first day of the following month.

\subsection{Data processing}\label{Subsec 5.2}

Theoretically, trading strategies vary continuously in time, while in the empirical analysis a daily trading frequency is used. The following procedure illustrates how we examine the gains and losses in our portfolio relative to the market portfolio.

We discretize the time horizon as $0=t_{0}<t_{1}<\cdots<t_{N-1}=T$, where $N$ is the total number of trading days.
\begin{itemize}
\item The transaction on day $t_{l}$, for all $l\in\{1,\cdots,N-1\}$, is made at the beginning of day ($\underline{t}_{l}$), taking the beginning of day $t_{l}$ market weights $\mu(\underline{t}_{l})$ as inputs. These market weights $\mu(\underline{t}_{l})$ are computed by
\[
\mu_{i}(\underline{t}_{l})=\frac{\mathrm{MV}_{i}(\underline{t}_{l})}{\Sigma(\underline{t}_{l})},\quad i\in\{1,\cdots,d\},
\]
where $\mathrm{MV}_{i}(\underline{t}_{l})$ is the market value of stock $i$ at the beginning of day $t_{l}$, which is assumed to be equal to the market value attainable at the end of the last trading day $t_{l-1}$, and $\Sigma(\underline{t}_{l})=\sum_{j=1}^{d}\mathrm{MV}_{j}(\underline{t}_{l})$ denotes the total market capitalization at the beginning of day $t_{l}$.
\\
\item The theoretical (non-self-financing) trading strategy throughout $t_{l}$, denoted by $\theta(t_{l})$, is computed based on either \eqref{eq 3.1} or \eqref{eq theta2}, taking $\mu(\underline{t}_{l})$ as inputs. Denote the implemented (self-financing) trading strategy corresponding to $\theta(t_{l})$ by $\phi(t_{l})$. Then $V^{\phi}(\underline{t}_{l})$, the beginning of day $t_{l}$ wealth of the portfolio corresponding to $\phi(t_{l})$, is given by
\begin{equation}\label{eq sf1}
V^{\phi}(\underline{t}_{1})=\frac{V^{\phi}(\overline{t}_{l-1})\Sigma(\overline{t}_{l-1})}
{\Sigma(\underline{t}_{1})}.
\end{equation}
This is based on the assumption that the real portfolio wealth does not change overnight. In \eqref{eq sf1}, $V^{\phi}(\overline{t}_{l-1})$ and $\Sigma(\overline{t}_{l-1})$ are the end of day $t_{l-1}$ portfolio wealth and total market capitalization, respectively, computed at $\overline{t}_{l-1}$ (thus already known at $\underline{t}_{l}$).
\\
\item To derive the implemented (self-financing) trading strategy $\phi(t_{l})$ corresponding to $\theta(t_{l})$, we compute the number
\begin{equation}\label{eq C}
C(t_{l})=\sum_{j=1}^{d}\theta_{j}(t_{l})\mu_{j}(\underline{t}_{l})-V^{\phi}(\underline{t}_{l}).
\end{equation}
Then $\phi(t_{l})$ is derived by
\begin{equation}\label{eq sf}
\phi_{i}(t_{l})=\theta_{i}(t_{l})-C(t_{l}),\quad i\in\{1,\cdots,d\}.
\end{equation}
This guarantees $V^{\phi}(\underline{t}_{l})=\sum_{i=1}^{d}\phi_{i}(t_{l})\mu_{i}(\underline{t}_{l})$.
\\
\item At the end of day $t_{l}$, the return indexes of the stocks for $t_{l}$ are available, and the total returns $\mathrm{TR}(t_{l})$ are computed through dividing the return indexes of $t_{l}$ with the return indexes of $t_{l-1}$. Then the end of day $t_{l}$ implied market values $\mathrm{MV}(\overline{t}_{l})$, which take the dividend payments into consideration, are given by
\[
\mathrm{MV}_{i}(\overline{t}_{l})=\mathrm{MV}_{i}(\underline{t}_{l})\mathrm{TR}_{i}(t_{l}),\quad i\in\{1,\cdots,d\}.
\]
The end of day $t_{l}$ modified total market capitalization $\Sigma(\overline{t}_{l})$ and market weights $\mu(\overline{t}_{l})$ are calculated similarly as $\Sigma(\underline{t}_{l})$ and $\mu(\underline{t}_{l})$, with $\mathrm{MV}(\underline{t}_{l})$ replaced by $\mathrm{MV}(\overline{t}_{l})$.
\\
\item The end of day $t_{l}$ portfolio wealth is then computed by
\[
V^{\phi}(\overline{t}_{l})=\sum_{j=1}^{d}\phi_{j}(t_{l})\mu_{j}(\overline{t}_{l}).
\]
Note that we have
\begin{equation}\label{eq V}
V^{\phi}(\overline{t}_{l})=V^{\phi}(\underline{t}_{l})+\sum_{j=1}^{d}\theta_{j}(t_{l})\left(\mu_{j}(\overline{t}_{l})-\mu_{j}(\underline{t}_{l})\right).
\end{equation}
\end{itemize}
In particular, at the beginning of day $t_{0}$, all of the above steps are still applied, except that we have $V^{\phi}(\underline{t}_{0})=1$ instead of \eqref{eq sf1} due to Definition~\ref{Definition realtive_arbitrage}.

\vspace{0.5cm}

\section{Examples and empirical results}\label{sec 6}

In this section, several examples of portfolio generating functions are empirically studied.

\vspace{0.5cm}

\begin{example}\label{Example 5.3}
Define the generalized \textit{entropy function}
\[
G(\lambda,x)=\lambda\sum_{i=1}^{d}x_{i}\log\left(\frac{1}{x_{i}}\right),\quad \lambda\in\mathbb{R}_{+},~x\in\Delta_{+}^{d},
\]
with values in $(0,\lambda\log d)$, for fixed $\lambda>0$. Suppose that $\mu(\cdot)$ takes values in $\Delta_{+}^{d}$ and that $\Lambda(\cdot)$ is $(0,\infty)$-valued.

From \eqref{eq 3.4} we have
\begin{equation}\label{eq gamma gef}
\Gamma^{G}(\cdot)=\sum_{i=1}^{d}\int_{0}^{\cdot}\mu_{i}(t)\log\mu_{i}(t)\mathrm{d}\Lambda(t)+\frac{1}{2}\sum_{i=1}^{d}\int_{0}^{\cdot}\Lambda(t)
\frac{\mathrm{d}\langle\mu_{i},\mu_{i}\rangle(t)}{\mu_{i}(t)}.
\end{equation}
Then $G$ is a Lyapunov function for $\Lambda(\cdot)$ and $\mu(\cdot)$ provided that $\Gamma^{G}(\cdot)$ is non-decreasing. One sufficient condition for this to hold is that $\Lambda(\cdot)$ is non-increasing.

From \eqref{eq 4.4}, the trading strategy $\boldsymbol{\varphi}(\cdot)$, generated additively by $G$, has components
\begin{equation}\label{eq gef_varphi}
\boldsymbol{\varphi}_{i}(\cdot)=\Gamma^{G}(\cdot)-\Lambda(\cdot)\log\mu_{i}(\cdot),\quad i\in\{1,\cdots,d\}.
\end{equation}
Using \eqref{eq 4.3}, the corresponding wealth process $V^{\boldsymbol{\varphi}}(\cdot)=G\big(\Lambda(\cdot),\mu(\cdot)\big)+\Gamma^{G}(\cdot)$ is strictly positive if $G$ is a generalized Lyapunov function.

For the multiplicative generation, $G$ is required to be bounded away from zero. One sufficient condition for this to hold is that $\Lambda(\cdot)$ is bounded away from $0$ and the market is diverse on $[0,\infty)$, i.e., there exists $\epsilon>0$ such that $G\big(\Lambda(t),\mu(t)\big)\geq \Lambda(t)\epsilon$, for all $t\geq0$ (see Proposition~2.3.2 in \cite{MR1894767}). Then from \eqref{eq 4.10}, the trading strategy $\boldsymbol{\psi}(\cdot)$, generated multiplicatively by $G$, has components
\[
\boldsymbol{\psi}_{i}(\cdot)=-\Lambda(\cdot)\log\mu_{i}(\cdot)\exp\left(\int_{0}^{\cdot}\frac{\mathrm{d}\Gamma^{G}(t)}{G\big(\Lambda(t),\mu(t)\big)}\right),\quad i\in\{1,\cdots,d\}.
\]
The corresponding wealth process $V^{\boldsymbol{\psi}}(\cdot)$ is given in \eqref{eq 4.9}.

Now, let us discuss sufficient conditions for the existence of arbitrage relative to the market. For the Lyapunov function $G$, let us consider
\begin{equation}\label{eq normalized G}
\mathcal{G}=\frac{G}{G\big(\Lambda(0),\mu(0)\big)}
\end{equation}
normalized to have initial value $1$, together with the non-decreasing process
\begin{equation}\label{eq Gamma^G1}
\Gamma^{\mathcal{G}}(\cdot)=\frac{\Gamma^{G}(\cdot)}{G\big(\Lambda(0),\mu(0)\big)}.
\end{equation}
Then from Theorem~\ref{Theorem 5.1}, if
\[
\mathrm{\textsf{P}}\big[\Gamma^{\mathcal{G}}(T_{\ast})>1\big]=\mathrm{\textsf{P}}\big[\Gamma^{G}(T_{\ast})>G\big(\Lambda(0),\mu(0)\big)\big]=1,
\]
then the trading strategy $\boldsymbol{\varphi}(\cdot)/G\big(\Lambda(0),\mu(0)\big)$, generated additively by $\mathcal{G}$, is strong relative arbitrage over every time horizon $[0,T]$ with $T\geq T_{\ast}$.

Similarly, from Theorem~\ref{Theorem 5.2}, if
\[
\mathrm{\textsf{P}}\big[\Gamma^{\mathcal{G}}(T_{\ast})>1+\varepsilon\big]=\mathrm{\textsf{P}}\big[\Gamma^{G}(T_{\ast})>G\big(\Lambda(0),\mu(0)\big)(1+\varepsilon)\big]=1,
\]
then the trading strategy $\boldsymbol{\psi}^{(c)}(\cdot)$, generated multiplicatively by
\begin{equation}\label{eq Gc}
G^{(c)}=\frac{G+c}{G\big(\Lambda(0),\mu(0)\big)+c},
\end{equation}
for some sufficiently large $c>0$, is strong relative arbitrage over every time horizon $[0,T]$ with $T\geq T_{\ast}$.

\begin{figure}[ht]
\centering
  \begin{minipage}[b]{0.49\textwidth}
    \includegraphics[width=\textwidth]{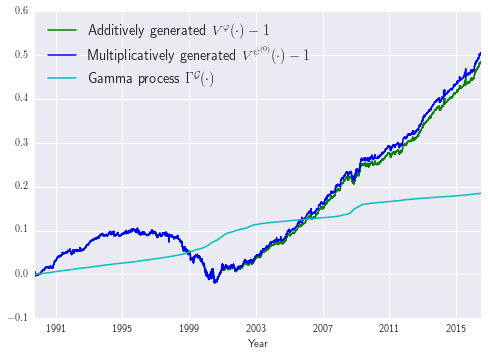}
    \caption{The case $\Lambda(\cdot)=1$.}
    \label{fg GEF1}
  \end{minipage}
\end{figure}

\begin{figure}[ht]
  \centering
  \begin{minipage}[b]{0.49\textwidth}
    \includegraphics[width=\textwidth]{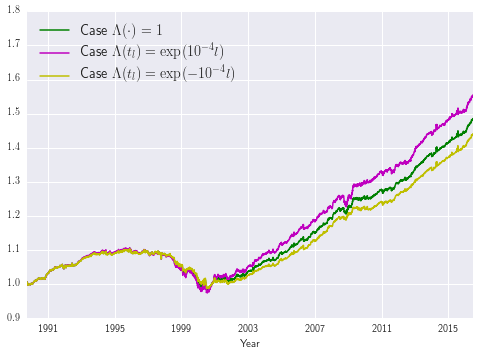}
    \caption{Wealth process $V^{\boldsymbol{\varphi}}(\cdot)$ with $\Lambda(\cdot)$ a deterministic exponential.\vspace{3mm}}
    \label{fg GEF2}
  \end{minipage}
  \hfill
  \begin{minipage}[b]{0.49\textwidth}
    \includegraphics[width=\textwidth]{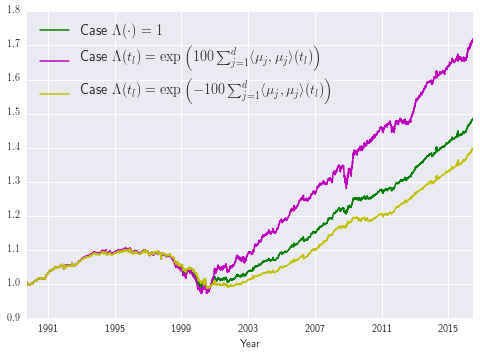}
    \caption{Wealth process $V^{\boldsymbol{\varphi}}(\cdot)$ with $\Lambda(\cdot)$ an exponential of the quadratic variation of $\mu(\cdot)$.}
    \label{fg GEF4}
  \end{minipage}
\end{figure}

\begin{figure}[ht]
  \centering
  \begin{minipage}[b]{0.49\textwidth}
    \includegraphics[width=\textwidth]{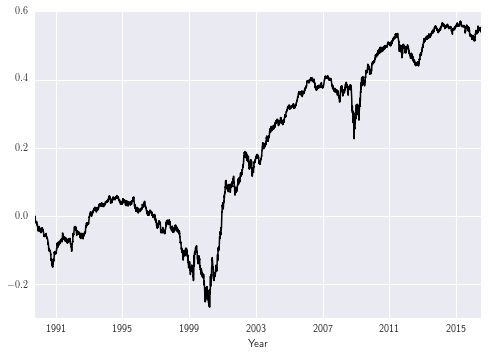}
    \caption{Integration process $E(\cdot)$ with components given by \eqref{eq sf4}.\vspace{3mm}}
    \label{fg GEF6}
  \end{minipage}
  \hfill
  \begin{minipage}[b]{0.49\textwidth}
    \includegraphics[width=\textwidth]{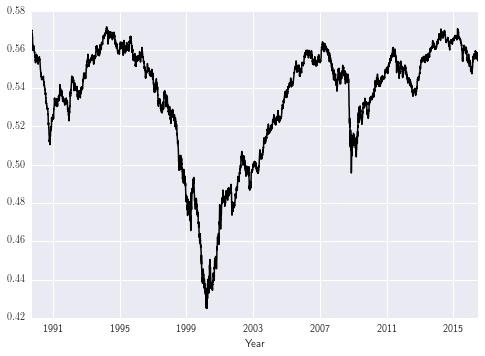}
    \caption{Process $\sum_{i=1}^{d}(\mu_{i}\wedge0.002)(\cdot)$ as a measure of the market diversification degree in the S\&P 500 market.}
    \label{fg GEF5}
  \end{minipage}
\end{figure}

Figure~\ref{fg GEF1} presents $\Gamma^{\mathbf{G}}(\cdot)$ given in \eqref{eq Gamma^G1} and the wealth processes $V^{\boldsymbol{\varphi}}(\cdot)$ and $V^{\boldsymbol{\psi}^{(0)}}(\cdot)$, with finite variation process $\Lambda(\cdot)=1$. As we can observe from the figure, both $V^{\boldsymbol{\varphi}}(\cdot)$ and $V^{\boldsymbol{\psi}^{(0)}}(\cdot)$ have been increasing since the year 2000.

Figure~\ref{fg GEF2} and Figure~\ref{fg GEF4} display the wealth processes $V^{\boldsymbol{\varphi}}(\cdot)$ corresponding to two different groups of $\Lambda(\cdot)$. Namely, for all $l\in\{1,\cdots,N\}$, in Figure~\ref{fg GEF2}, the wealth processes $V^{\boldsymbol{\varphi}}(\cdot)$ corresponding to $\Lambda(t_{l})=\exp\left(10^{-4}l\right)$ and $\Lambda(t_{l})=\exp\left(-10^{-4}l\right)$ are plotted; in Figure~\ref{fg GEF4}, the wealth processes $V^{\boldsymbol{\varphi}}(\cdot)$ corresponding to
\[
\Lambda(t_{l})=\exp\left(100\sum_{j=1}^{d}\langle\mu_{j},\mu_{j}\rangle(t_{l})\right)\quad\text{and}\quad \Lambda(t_{l})=\exp\left(-100\sum_{j=1}^{d}\langle\mu_{j},\mu_{j}\rangle(t_{l})\right)
\]
are plotted. The constants $10^{-4}$ and $100$ are chosen such that, with these forms, the daily changes of both $\mathcal{G}\big(\Lambda(\cdot),\mu(\cdot)\big)$ and $\Gamma^{\mathcal{G}}(\cdot)$ are roughly at the same level of magnitude. Hence, in \eqref{eq 4.3}, neither part on the right hand side dominates the other.

As we can observe from the figures, choosing $\Lambda(\cdot)$ increasing seems to lead to a better performance than choosing $\Lambda(\cdot)$ constant, which again seems to be better than choosing $\Lambda(\cdot)$ decreasing. We attribute the reason behind this observation to the state of market diversification as follows.

Observe that \eqref{eq V} yields
\begin{equation}\label{eq sf3}
V^{\boldsymbol{\varphi}}(\overline{t}_{l})=V^{\boldsymbol{\varphi}}(\underline{t}_{l})+\frac{1}{G\big(\Lambda(0),\mu(\underline{0})\big)}\Lambda(t_{l})D(t_{l}),\quad l\in\{0,\cdots,N\},
\end{equation}
where $D(t_{l})$ is given by
\begin{equation}\label{eq sf4}
D(t_{l})=\sum_{j=1}^{d}-\log\mu_{j}(\underline{t}_{l})\left(\mu_{j}(\overline{t}_{l})-\mu_{j}(\underline{t}_{l})\right).
\end{equation}

The value $D(t_{l})$ can be considered as an indicator of the direction of changes in market weights from the beginning to the end of date $t_{l}$. The value $D(t_{l})$ will be positive (negative), if market weights are shifted from companies with large (small) beginning of day market weights to companies with small (large) beginning of day market weights throughout date $t_{l}$. We consider a simple example to better understand why this is the case.

Fix $d=2$ and assume that $\mu_{1}(\underline{t}_{l})>\mu_{2}(\underline{t}_{l})$. Then
\[
\begin{aligned}
D(t_{l})&=-\log\mu_{1}(\underline{t}_{l})\left(\mu_{1}(\overline{t}_{l})-\mu_{1}(\underline{t}_{l})\right)
-\log\mu_{2}(\underline{t}_{l})\left(\mu_{2}(\overline{t}_{l})-\mu_{2}(\underline{t}_{l})\right)\\
&=\left(-\log\mu_{1}(\underline{t}_{l})+\log\mu_{2}(\underline{t}_{l})\right)\left(\mu_{1}(\overline{t}_{l})-\mu_{1}(\underline{t}_{l})\right)
\end{aligned}
\]
holds due to the fact that $\left(\mu_{1}(\overline{t}_{l})-\mu_{1}(\underline{t}_{l})\right)=-\left(\mu_{2}(\overline{t}_{l})-\mu_{2}(\underline{t}_{l})\right)$. Hence, $D(t_{l})>0$ if and only if $\mu_{1}(\overline{t}_{l})<\mu_{1}(\underline{t}_{l})$, i.e., the market weight of the company with larger beginning of day market weight decreases, while the market weight of the company with smaller beginning of day market weight increases.

Hence, a positive $D(\cdot)$ indicates an enhancement in market diversification, while $D(\cdot)$ being negative actually implies a reduction in market diversification. Figure~ \ref{fg GEF6} plots the cumulative process $E(\cdot)=\sum_{t_{l}=t_{1}}^{\cdot}D(t_{l})$. The process $E(\cdot)$ is increasing (decreasing) whenever $D(\cdot)$ is positive (negative). From Figure~\ref{fg GEF6} we can observe that after a slight increase from the year 1991 to the year 1995, $E(\cdot)$ keeps declining till the year 2000. Then $E(\cdot)$ rises up in the long run from the year 2000 until now.

The behavior of the process $E(\cdot)$ is in line with another measurement of the market diversification. More precisely, let us consider the process $\sum_{i=1}^{d}(\mu_{i}\wedge0.002)(\cdot)$. Note that the value $0.002=1/500$, which is roughly the number of constituents in the portfolio. This process is a measure of the market diversification, as it goes up when the market weights of small companies become larger, i.e., the market diversification is strengthened. Figure~\ref{fg GEF5} plots the process, which first grows from the year 1991 to the year 1995. Then from the year 1995 to 2000, the process declines rapidly. This indicates that during this period, the market diversification weakens. On the contrary, the market diversification strengthens afterwards until the year 2008, as the process goes up. Then the level of market diversification remains within a relatively small range.

As a result, according to \eqref{eq sf3}, if the market presents a trend of increasing diversification, an increasing positive $\Lambda(\cdot)$ helps to reinforce this effect, and further assists in pulling up $V^{\boldsymbol{\varphi}}(\cdot)$, while a decreasing positive $\Lambda(\cdot)$ is counteractive. On the other hand, if the market presents a trend of decreasing diversification, then a decreasing positive $\Lambda(\cdot)$ helps to slow down the declining speed of $V^{\boldsymbol{\varphi}}(\cdot)$, while an increasing positive $\Lambda(\cdot)$ would make the speed even faster. This is confirmed in Figure~\ref{fg GEF2} and Figure~\ref{fg GEF4}, as from the year 1991 to the year 1995 and from the year 2000 till now, an increasing positive $\Lambda(\cdot)$ makes $V^{\boldsymbol{\varphi}}(\cdot)$ perform better, while from the year 1995 to the year 2000, $V^{\boldsymbol{\varphi}}(\cdot)$ corresponding to a decreasing positive $\Lambda(\cdot)$ is slightly larger.

Although an increasing positive $\Lambda(\cdot)$ has positive effect on the portfolio performance $V^{\boldsymbol{\varphi}}(\cdot)$ whenever the market diversification strengthens, we are not allowed to choose $\Lambda(\cdot)$ arbitrarily fast increasing. The reason is that $\Gamma^{\mathcal{G}}(\cdot)$ will become negative and decrease rapidly if the increments in $\Lambda(\cdot)$ are sufficiently large, which will result in a negative $\boldsymbol{\psi}(\cdot)$ given in \eqref{eq gef_varphi}.

As for the multiplicative generation, the different choices of finite variation processes do not change the wealth processes significantly. Indeed, according to \eqref{eq gamma gef}, an increasing $\Lambda(\cdot)$ may slow down the growth rate of $\Gamma(\cdot)$, or even turn $\Gamma(\cdot)$ into a decreasing one. When applying \eqref{eq sf} to $\widetilde{\boldsymbol{\vartheta}}(\cdot)$ from \eqref{eq theta2}, we have
\[
V^{\boldsymbol{\psi}^{(c)}}(\overline{t}_{l})=\exp\left(\int_{0}^{\underline{t}_{l}}
\frac{\mathrm{d}\Gamma^{G}(t)}{G\big(\Lambda(t),\mu(t)\big)+c}\right)
\frac{\Lambda(t_{l})}{G\big(\Lambda(0),\mu(\underline{0})\big)+c}D(t_{l})+V^{\boldsymbol{\psi}^{(c)}}(\underline{t}_{l}),
\]
for all $l\in\{0,\cdots,N\}$, with $D(\cdot)$ given in \eqref{eq sf4}. In this example, according to the above equation, the positive effect in boosting $V^{\boldsymbol{\psi}^{(c)}}(\cdot)$ contributed by an increasing positive $\Lambda(\cdot)$ is counteracted more or less by the opposite impact the same $\Lambda(\cdot)$ has on the exponential part. A similar analysis also applies to a decreasing positive $\Lambda(\cdot)$. Therefore, under the above mentioned situation (market diversification increases in general), the different choices of a monotone $\Lambda(\cdot)$ do not influence $V^{\boldsymbol{\psi}^{(c)}}(\cdot)$ as much as they do on $V^{\boldsymbol{\varphi}}(\cdot)$.

Note that our process $D(\cdot)$ is related but not the same as the Bregman divergence
\[
D_{B,G}\left[\mu(\overline{t}_{l})|\mu(\underline{t}_{l})\right]=\Lambda(t_{l})D(t_{l})
-\left(G\big(\Lambda(t_{l}),\mu(\overline{t}_{l})\big)-G\big(\Lambda(t_{l}),\mu(\underline{t}_{l})\big)\right),
\]
defined in Definition~3.6 of \cite{Wong_2017}. For its connection to optimal transport, we refer to \cite{Wong_2017}.
\qed
\end{example}

The following example is motivated by \cite{Schied_2016}.
\vspace{5mm}
\begin{example}\label{Example FDW}
Consider the function
\[
G(\lambda,x)=\left(\sum_{i=1}^{d}\left(\alpha x_{i}+(1-\alpha)\lambda_{i}\right)^{p}\right)^{\frac{1}{p}},\quad \lambda\in\mathbb{R}_{+}^{d},~x\in\Delta_{+}^{d},
\]
with constants $\alpha,p\in(0,1)$. Then $G$ is concave.

For fixed constant $\delta>0$, define the $\mathbb{R}_{+}^{d}$-valued moving average process $\Lambda(\cdot)$ by
\[
\begin{aligned}
\Lambda_{i}(\cdot)=
\begin{cases}
\frac{1}{\delta}\int_{0}^{\cdot}\mu_{i}(t)\mathrm{d}t+\frac{1}{\delta}\int_{\cdot-\delta}^{0}\mu_{i}(0)\mathrm{d}t\quad&\text{on~}[0,\delta)\\
\frac{1}{\delta}\int_{\cdot-\delta}^{\cdot}\mu_{i}(t)\mathrm{d}t\quad&\text{on~}[\delta,\infty)
\end{cases},
\end{aligned}
\]
for all $i\in\{1,\cdots,d\}$.

Write $\overline{\mu}(\cdot)=\alpha\mu(\cdot)+(1-\alpha)\Lambda(\cdot)$. Then by \eqref{eq 3.4},
\[
\begin{aligned}
\Gamma^{G}(\cdot)&=-(1-\alpha)\sum_{i=1}^{d}\int_{0}^{\cdot}\left(\frac{G\big(\Lambda(t),\mu(t)\big)}{\overline{\mu}_{i}(t)}\right)^{1-p}\mathrm{d}\Lambda_{i}(t)\\
&-\frac{\alpha^{2}(1-p)}{2}\sum_{i,j=1}^{d}\int_{0}^{\cdot}\left(\frac{G\big(\Lambda(t),\mu(t)\big)}{\overline{\mu}_{i}(t)\overline{\mu}_{j}(t)}\right)^{1-p}
\frac{1}{\sum_{v=1}^{d}\left(\overline{\mu}_{v}(t)\right)^{p}}
\mathrm{d}\langle\mu_{i},\mu_{j}\rangle(t)\\
&+\frac{\alpha^{2}(1-p)}{2}\sum_{i=1}^{d}\int_{0}^{\cdot}\left(\frac{G\big(\Lambda(t),\mu(t)\big)}{\overline{\mu}_{i}(t)}\right)^{1-p}
\frac{1}{\overline{\mu}_{i}(t)}\mathrm{d}\langle\mu_{i},\mu_{i}\rangle(t).
\end{aligned}
\]
Notice that $G$ is not Lyapunov in general.

The trading strategies $\boldsymbol{\varphi}(\cdot)$ and $\boldsymbol{\psi}(\cdot)$, generated additively and multiplicatively by $G$, respectively, are given by
\[
\boldsymbol{\varphi}_{i}(\cdot)= G\big(\Lambda(\cdot),\mu(\cdot)\big)\left(\frac{\alpha\left(\overline{\mu}_{i}(\cdot)\right)^{p}}{\overline{\mu}_{i}(\cdot)\sum_{v=1}^{d}\left(\overline{\mu}_{v}(\cdot)\right)^{p}}-
\sum_{j=1}^{d}\frac{\alpha\mu_{j}(\cdot)\left(\overline{\mu}_{j}(\cdot)\right)^{p}}{\overline{\mu}_{j}(\cdot)\sum_{v=1}^{d}\left(\overline{\mu}_{v}(\cdot)\right)^{p}}+1\right)
+\Gamma^{G}(\cdot)
\]
and
\[
\boldsymbol{\psi}_{i}(\cdot)=\left(\boldsymbol{\varphi}_{i}(\cdot)-\Gamma^{G}(\cdot)\right)
\exp\left(\int_{0}^{\cdot}\frac{\mathrm{d}\Gamma^{G}(t)}{G\big(\Lambda(t),\mu(t)\big)}\right),
\]
for all $i\in\{1,\cdots,d\}$. The corresponding wealth processes $V^{\boldsymbol{\varphi}}(\cdot)$ and $V^{\boldsymbol{\psi}}(\cdot)$ can be derived from \eqref{eq 4.3} and \eqref{eq 4.9}, respectively.

Consider the normalized regular function $\mathcal{G}$ given in \eqref{eq normalized G} and the corresponding process $\Gamma^{\mathcal{G}}(\cdot)$ given in \eqref{eq Gamma^G1}. By Theorem~\ref{Theorem 5.2}, if
\[
\mathrm{\textsf{P}}\big[\Gamma^{\mathcal{G}}(T_{\ast})>1+\varepsilon\big]=\mathrm{\textsf{P}}\big[\Gamma^{G}(T_{\ast})>G\big(\Lambda(0),\mu(0)\big)(1+\varepsilon)\big]=1,
\]
then the trading strategy $\boldsymbol{\psi}^{(c)}(\cdot)$, generated multiplicatively by $G^{(c)}$ given in \eqref{eq Gc} for some sufficiently large $c>0$, is strong relative arbitrage over the time horizon $[0,T_{\ast}]$.

\begin{figure}[ht]
  \centering
  \begin{minipage}[b]{0.49\textwidth}
    \includegraphics[width=\textwidth]{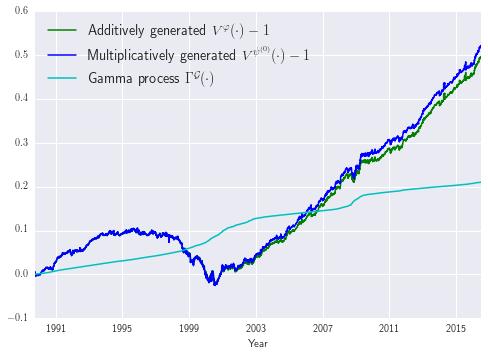}
    \caption{The case $\delta=250$ days, $p=0.8$ and $\alpha=1$.}
    \label{fg MA1}
  \end{minipage}
  \hfill
  \begin{minipage}[b]{0.49\textwidth}
    \includegraphics[width=\textwidth]{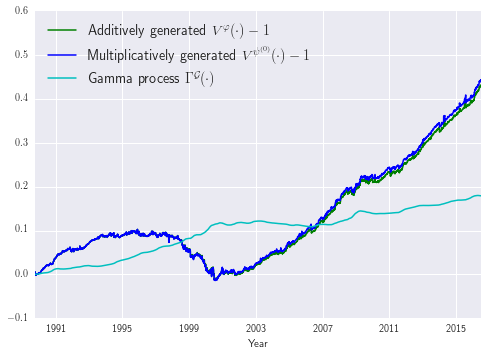}
    \caption{The case $\delta=250$ days, $p=0.8$ and $\alpha=0.6$.}
    \label{fg MA2}
  \end{minipage}
\end{figure}

To simulate the relative performance of the portfolio, we use the parameters $\delta=250$ days and $p=0.8$. Figure~\ref{fg MA1} shows $\Gamma^{\mathcal{G}}(\cdot)$ and the wealth processes $V^{\boldsymbol{\varphi}}(\cdot)$ and $V^{\boldsymbol{\psi}^{(0)}}(\cdot)$ without the effect of the moving average part, i.e., $\alpha=1$. In this case, $\mathcal{G}$ is Lyapunov. The relative performance of the portfolio is similar to that in Example~\ref{Example 5.3}, when the finite variation process is chosen to be constant. Figure~\ref{fg MA2} presents the case when $\alpha=0.6$. It can be observed that $\Gamma^{\mathcal{G}}(\cdot)$ increases slower when the moving average part is considered. Compared with the case that the moving average part is not included, the wealth processes $V^{\boldsymbol{\varphi}}(\cdot)$ and $V^{\boldsymbol{\psi}^{(0)}}(\cdot)$ also take smaller values in the long run. This is due to the fact that when $\alpha$ decreases, the volatility of $\overline{\mu}(\cdot)$ decreases as well. In this case, we trade slower, and the gains and losses will also be relatively less.
\qed
\end{example}

\vspace{5mm}

\section{Conclusion}\label{sec 7}

\cite{MR3663643} build a simple and intuitive structure by interpreting the portfolio generating functions $G$ initiated by \cite{F_generating, MR1861997, MR1894767} as Lyapunov functions. They formulate conditions for the existence of strong arbitrage relative to the market over appropriate time horizons. The purpose of this paper is to investigate the dependence of the portfolio generating functions $G$ on an extra $\mathbb{R}^{m}$-valued, progressive, continuous process $\Lambda(\cdot)$ of finite variation on $[0,T]$, for all $T\geq0$.

The results of this paper are illuminated by several examples and shown to work on empirical data using stocks from the S\&P 500 index. The effects that different choices of $\Lambda(\cdot)$ have on the portfolio wealths are analyzed. Provided that the market undergos an explicit trend of either increasing or decreasing market diversification, certain choices of $\Lambda(\cdot)$ are better than others.

\vspace{0.5cm}

\appendix

\section{Proofs of Theorems \ref{Theorem new1} and \ref{Theorem new2}}\label{sec A}

\subsection{Preliminaries}

Before providing the proof of Theorem~\ref{Theorem new1}, we discuss some technical details.

Recall the open set $\mathcal{W}$ from \eqref{eq W} and consider a continuous function $g:\mathcal{W}\rightarrow\mathbb{R}$. Define a function $\overline{g}:\mathbb{R}^{m+d}\rightarrow\mathbb{R}$ by
\[
\begin{aligned}
\overline{g}(z)=
\begin{cases}
g(z),&\quad\text{if}~z\in\mathcal{W}\\
0,&\quad\text{if}~z\notin\mathcal{W}
\end{cases}.
\end{aligned}
\]
Next, consider the family of functions $(g_{n_{1},n_{2}})_{n_{1},n_{2}\in\mathbb{N}}$ with $g_{n_{1},n_{2}}:\mathcal{W}\rightarrow\mathbb{R}$ given by
\begin{equation}\label{eq mollification 1}
g_{n_{1},n_{2}}(\lambda,x)=\int_{\mathbb{R}^{d}}\eta_{n_{2}}(y)\int_{\mathbb{R}^{m}}\eta_{n_{1}}(u)\overline{g}(\lambda-u,x-y)\mathrm{d}u\mathrm{d}y,
\end{equation}
for all $(\lambda,x)\in\mathcal{W}$, with $g_{n_{1},n_{2}}(\lambda,x)=0$ whenever the right hand side of \eqref{eq mollification 1} is not defined. Here in \eqref{eq mollification 1}, for $z\in\mathbb{R}^{l}$ and $n\in\mathbb{N}$,
\begin{equation}\label{eq etan}
\begin{aligned}
\eta_{n}(z)=
\begin{cases}
\beta n^{l}\exp\left(\frac{1}{n^{2}\|z\|^{2}_{2}-1}\right),\quad&\text{if}~\|z\|_{2}<\frac{1}{n}\\
0,&\text{if}~\|z\|_{2}\geq\frac{1}{n}
\end{cases}
\end{aligned}
\end{equation}
is used with the normalization constant
\[
\beta=\left(\int_{\mathbb{R}^{l}}\exp\left(\frac{1}{\|y\|^{2}_{2}-1}\right)\mathrm{d}y\right)^{-1},
\]
independent of $n$.
\\
\begin{lemma}\label{Lemma 1}
Let $\mathcal{\overline{V}}$ denote any closed subset of $\mathcal{W}$. Consider a continuous function $g:\mathcal{W}\rightarrow\mathbb{R}$ and the mollification $(g_{n_{1},n_{2}})_{n_{1},n_{2}\in\mathbb{N}}$ of $g$ defined as in \eqref{eq mollification 1}.
\begin{enumerate}[label=(\roman*)]
  \item We have
  \[
  \lim_{n_{2}\uparrow\infty}\lim_{n_{1}\uparrow\infty}g_{n_{1},n_{2}}=g.
  \]
  \item For $n_{1},n_{2}\in\mathbb{N}$ large enough, $g_{n_{1},n_{2}}\in C^{\infty}(\mathcal{\overline{V}})$.
  \item If there exists a constant $L=L(\mathcal{\overline{V}})\geq0$ such that, for all $(\lambda_{1},x),(\lambda_{2},x)\in\mathcal{\overline{V}}$,
\[
\left|g(\lambda_{1},x)-g(\lambda_{2},x)\right|\leq L\|\lambda_{1}-\lambda_{2}\|_{2},
\]
then, for $n_{1},n_{2}\in\mathbb{N}$ large enough and all $(\lambda,x)\in\mathcal{\overline{V}}$, we have
\[
\left|\frac{\partial g_{n_{1},n_{2}}}{\partial \lambda_{v}}(\lambda,x)\right|\leq L,\quad v\in\{1,\cdots,m\}.
\]
\item If $g\in C^{0,1}$, then, for all $(\lambda,x)\in\mathcal{W}$, we have
\[
\lim_{n_{2}\uparrow\infty}\lim_{n_{1}\uparrow\infty}\frac{\partial g_{n_{1},n_{2}}}{\partial x_{i}}(\lambda,x)=\frac{\partial g}{\partial x_{i}}(\lambda,x),\quad i\in\{1,\cdots,d\}.
\]
\item If $g\in C^{0,1}$ and if there exists a constant $L=L(\mathcal{\overline{V}})\geq0$ such that, for all $(\lambda,x_{1}),(\lambda,x_{2})\in\mathcal{\overline{V}}$,
\[
\left\|\frac{\partial g}{\partial x}(\lambda,x_{1})-\frac{\partial g}{\partial x}(\lambda,x_{2})\right\|_{2}\leq L\|x_{1}-x_{2}\|_{2},
\]
then, for $n_{1},n_{2}\in\mathbb{N}$ large enough and all $(\lambda,x)\in\mathcal{\overline{V}}$, we have
\[
\left|\frac{\partial^{2}g_{n_{1},n_{2}}}{\partial x_{i}\partial x_{j}}(\lambda,x)\right|\leq L,\quad i,j\in\{1,\cdots,d\}.
\]
\end{enumerate}
\end{lemma}

\begin{proof}
For (i) and (ii), see Theorem~6 in Appendix~C of \cite{MR1625845}.

For (iii), observe that, for each $n_{1},n_{2}\in\mathbb{N}$ large enough and all $v\in\{1,\cdots,m\}$, \eqref{eq mollification 1} yields
\[
\begin{aligned}
&\left|\frac{\partial g_{n_{1},n_{2}}}{\partial \lambda_{v}}(\lambda,x)\right|=\left|\lim_{\delta\rightarrow0}\frac{g_{n_{1},n_{2}}(\lambda+\delta\mathbf{e}_{v},x)-g_{n_{1},n_{2}}(\lambda,x)}{\delta}\right|&\\
&=\left|\lim_{\delta\rightarrow0}\frac{1}{\delta}\int_{\mathbb{R}^{d}}\eta_{n_{2}}(y)\int_{\mathbb{R}^{m}}\eta_{n_{1}}(u)
\big(\overline{g}(\lambda+\delta\mathbf{e}_{v}-u,x-y)-\overline{g}(\lambda-u,x-y)\big)\mathrm{d}u\mathrm{d}y\right|&\\
&\leq\lim_{\delta\rightarrow0}\frac{1}{\delta}\int_{\mathbb{R}^{d}}\eta_{n_{2}}(y)\int_{\mathbb{R}^{m}}\eta_{n_{1}}(u)
\left|\overline{g}(\lambda+\delta\mathbf{e}_{v}-u,x-y)-\overline{g}(\lambda-u,x-y)\right|\mathrm{d}u\mathrm{d}y&\\
&\leq\lim_{\delta\rightarrow0}\frac{1}{\delta}\delta L\int_{\mathbb{R}^{d}}\eta_{n_{2}}(y)\int_{\mathbb{R}^{m}}\eta_{n_{1}}(u)\mathrm{d}u\mathrm{d}y=L,&
\end{aligned}
\]
for all $(\lambda,x)\in\mathcal{\overline{V}}$, where $\mathbf{e}_{v}$ is the unit vector in the $v$-th dimension.

For (iv), apply the dominated convergence theorem and (i) to $\frac{\partial g}{\partial x_{i}}$, for all $i\in\{1,\cdots,d\}$.

For (v), similarly to (iii), for each $n_{1},n_{2}\in\mathbb{N}$ large enough and all $i,j\in\{1,\cdots,d\}$, we have
\[
\begin{aligned}
&\left|\frac{\partial^{2}g_{n_{1},n_{2}}}{\partial x_{i}\partial x_{j}}(\lambda,x)\right|=\left|\lim_{\delta\rightarrow0}\frac{\frac{\partial g_{n_{1},n_{2}}}{\partial x_{i}}(\lambda,x+\delta\mathbf{e}_{j})-\frac{\partial g_{n_{1},n_{2}}}{\partial x_{i}}(\lambda,x)}{\delta}\right|&\\
&=\left|\lim_{\delta\rightarrow0}\frac{1}{\delta}\int_{\mathbb{R}^{d}}\eta_{n_{2}}(y)\int_{\mathbb{R}^{m}}\eta_{n_{1}}(u)
\left(\frac{\partial\overline{g}}{\partial x_{i}}(\lambda-u,x+\delta\mathbf{e}_{j}-y)-\frac{\partial\overline{g}}{\partial x_{i}}(\lambda-u,x-y)\right)\mathrm{d}u\mathrm{d}y\right|&\\
&\leq L,&
\end{aligned}
\]
for all $(\lambda,x)\in\mathcal{\overline{V}}$, where for the second equality we apply the dominated convergence theorem.
\end{proof}

The following lemma is an extension of Lemma~2 in \cite{MR610155}. For a continuous function $g:\mathcal{W}\rightarrow\mathbb{R}$, consider its corresponding mollification $(g_{n_{1},n_{2}})_{n_{1},n_{2}\in\mathbb{N}}$ defined as in \eqref{eq mollification 1}.
\\
\begin{lemma}\label{Lemma 4}
If a continuous function $g:\mathcal{W}\rightarrow\mathbb{R}$ is concave in its second argument, then
\[
\lim_{n_{2}\uparrow\infty}\lim_{n_{1}\uparrow\infty}\frac{\partial g_{n_{1},n_{2}}}{\partial x_{i}}=f_{i},\quad i\in\{1,\cdots,d\},
\]
for some measurable function $f_{i}:\mathcal{W}\rightarrow\mathbb{R}$, bounded on any compact $\mathcal{\overline{V}}\subset\mathcal{W}$.
\end{lemma}

\begin{proof}
With the notation in \eqref{eq etan}, we have
\[
\eta_{n}(z)=n^{l}\eta_{1}(nz),\quad z\in\mathbb{R}^{l},~n\in\mathbb{N}.
\]
For $(\lambda,x)\in\mathcal{W}$ and $n_{2}\in\mathbb{N}$ large enough, the definition of $g_{n_{1},n_{2}}$ in \eqref{eq mollification 1}, the dominated convergence theorem, and Lemma~\ref{Lemma 1} yield
\[
\begin{aligned}
\lim_{n_{1}\uparrow\infty}\frac{\partial g_{n_{1},n_{2}}}{\partial x_{i}}(\lambda,x)&=\lim_{n_{1}\uparrow\infty}\int_{\mathbb{R}^{d}}\frac{\partial\eta_{n_{2}}}{\partial x_{i}}(x-y)\int_{\mathbb{R}^{m}}\eta_{n_{1}}(u)\overline{g}(\lambda-u,y)\mathrm{d}u\mathrm{d}y\\
&=\int_{\mathbb{R}^{d}}\frac{\partial\eta_{n_{2}}}{\partial x_{i}}(x-y)\lim_{n_{1}\uparrow\infty}\int_{\mathbb{R}^{m}}\eta_{n_{1}}(u)\overline{g}(\lambda-u,y)\mathrm{d}u\mathrm{d}y\\
&=\int_{\mathbb{R}^{d}}\frac{\partial\eta_{n_{2}}}{\partial x_{i}}(x-y)\overline{g}(\lambda,y)\mathrm{d}y\\
&=-\int_{\mathbb{R}^{d}}\frac{\partial\eta_{n_{2}}}{\partial y_{i}}(y)
\overline{g}(\lambda,x-y)\mathrm{d}y\\
&=\int_{\mathbb{R}^{d}}n_{2}\frac{\partial\eta_{1}}{\partial y_{i}}(y)
\overline{g}\left(\lambda,x+\frac{y}{n_{2}}\right)\mathrm{d}y\\
&=\int_{\mathbb{R}^{d}}\frac{\partial\eta_{1}}{\partial y_{i}}(y)
n_{2}\left(\overline{g}\left(\lambda,x+\frac{y}{n_{2}}\right)-\overline{g}\left(\lambda,x\right)\right)\mathrm{d}y.\\
\end{aligned}
\]
Note that the last equality holds due to the fact that
\[
\int_{\mathbb{R}^{d}}\frac{\partial\eta_{1}}{\partial y_{i}}(y)\mathrm{d}y=0.
\]

Next, for all $(\lambda,x)\in\mathcal{W}$ and $y\in\mathbb{R}^{d}$, define the one-sided directional partial derivative as
\[
\nabla g(\lambda,x;y)=\lim_{n_{2}\uparrow\infty}\frac{g\left(\lambda,x+y/n_{2}\right)-g(\lambda,x)}{1/n_{2}}.
\]
Such $\nabla g$ exists according to Theorem~23.1 in \cite{MR0274683}. Since $g$ is concave in the second argument, it is locally Lipschitz in its second argument on $\mathcal{W}$ (see Theorem~10.4 in \cite{MR0274683}). Hence, for each compact $\mathcal{\overline{V}}\subset\mathcal{W}$, there exists a constant $L=L(\mathcal{\overline{V}})\geq0$ such that $\nabla g(\lambda,x;y)\leq L$, for all $y\in\mathbb{R}^{d}$ and $(\lambda,x)$ in the interior of $\mathcal{\overline{V}}$.

The statement now follows with
\[
f_{i}(\lambda,x)=\int_{\mathbb{R}^{d}}\nabla\overline{g}(\lambda,x;y)\frac{\partial\eta_{1}}{\partial y_{i}}(y)\mathrm{d}y,
\]
for all $(\lambda,x)\in\mathcal{W}$, by the dominated convergence theorem.
\end{proof}

\begin{lemma}\label{Lemma 2}
Assume that $\mu(\cdot)$ has Doob-Meyer decomposition $\mu(\cdot)=\mu(0)+M(\cdot)+V(\cdot)$, where $M(\cdot)$ is a $d$-dimensional continuous local martingale and $V(\cdot)$ is a $d$-dimensional finite variation process with $M(0)=V(0)=0$. Moreover, suppose that,
\begin{enumerate}[label=(\roman*)]
\item for some open $\mathcal{V}\subset\mathcal{W}$, we have $\big(\Lambda(\cdot),\mu(\cdot)\big)=\big(\Lambda(\cdot\wedge\tau),\mu(\cdot\wedge\tau)\big)$, where
\[
\tau=\inf\left\{t\geq0;~\big(\Lambda(t),\mu(t)\big)\notin\mathcal{V}\right\};
\]
\item for some constant $\kappa\geq0$, we have
\begin{equation}\label{eq assumption 1}
\sum_{i=1}^{d}\left(\langle M_{i},M_{i}\rangle(\infty)+\int_{0}^{\infty}\mathrm{d}|V_{i}(t)|\right)+\sum_{v=1}^{m}\int_{0}^{\infty}\mathrm{d}|\Lambda_{v}(t)|\leq\kappa<\infty.
\end{equation}
\end{enumerate}

Consider two families of bounded functions $(h_{i})_{i\in\{1,\cdots,d\}}$ and $(h_{i}^{n_{1},n_{2}})_{n_{1},n_{2}\in\mathbb{N},i\in\{1,\cdots,d\}}$ with $h_{i},h_{i}^{n_{1},n_{2}}:\mathcal{V}\rightarrow\mathbb{R}$. If
\[
\lim_{n_{2}\uparrow\infty}\lim_{n_{1}\uparrow\infty}h_{i}^{n_{1},n_{2}}=h_{i},\quad i\in\{1,\cdots,d\},
\]
then there exist two random subsequences $\left(n_{1}^{k}\right)_{k\in\mathbb{N}}$ and $\left(n_{2}^{k}\right)_{k\in\mathbb{N}}$ with $\lim_{k\uparrow\infty}n_{1}^{k}=\infty=\lim_{k\uparrow\infty}n_{2}^{k}$ such that
\[
\lim_{k\uparrow\infty}\int_{0}^{t}\sum_{i=1}^{d}
h_{i}^{n_{1}^{k},n_{2}^{k}}\big(\Lambda(u),\mu(u)\big)\mathrm{d}\mu_{i}(u)=\int_{0}^{t}\sum_{i=1}^{d}h_{i}\big(\Lambda(u),\mu(u)\big)\mathrm{d}\mu_{i}(u),\quad\text{a.s.},
\]
for all $t\geq0$.
\end{lemma}

\begin{proof}
Fix $i\in\{1,\cdots,d\}$ and write
\[
\Theta_{i}^{n_{1},n_{2}}(\cdot)=h_{i}^{n_{1},n_{2}}\big(\Lambda(\cdot),\mu(\cdot)\big)-h_{i}\big(\Lambda(\cdot),\mu(\cdot)\big).
\]
By \eqref{eq assumption 1} and the bounded convergence theorem, we have
\[
\lim_{n_{2}\uparrow\infty}\lim_{n_{1}\uparrow\infty}\int_{0}^{\infty}\big(\Theta_{i}^{n_{1},n_{2}}(t)\big)^{2}\mathrm{d}\langle M_{i},M_{i}\rangle(t)=0,\quad\text{a.s.},
\]
and
\[
\lim_{n_{2}\uparrow\infty}\lim_{n_{1}\uparrow\infty}\left(\int_{0}^{\infty}\left|\Theta_{i}^{n_{1},n_{2}}(t)\right|\mathrm{d}|V_{i}(t)|\right)^{2}=0,\quad\text{a.s.}
\]
Therefore, by the bounded convergence theorem again, we have
\[
\begin{aligned}
0&=\mathbb{E}\left[\lim_{n_{2}\uparrow\infty}\lim_{n_{1}\uparrow\infty}\int_{0}^{\infty}\big(\Theta_{i}^{n_{1},n_{2}}(t)\big)^{2}\mathrm{d}\langle M_{i},M_{i}\rangle(t)\right]\\
&=\lim_{n_{2}\uparrow\infty}\lim_{n_{1}\uparrow\infty}\mathbb{E}\left[\int_{0}^{\infty}\big(\Theta_{i}^{n_{1},n_{2}}(t)\big)^{2}\mathrm{d}\langle M_{i},M_{i}\rangle(t)\right]\\
&=\lim_{n_{2}\uparrow\infty}\lim_{n_{1}\uparrow\infty}\mathbb{E}\left[\left(\int_{0}^{\infty}\Theta_{i}^{n_{1},n_{2}}(t)\mathrm{d}M_{i}(t)\right)^{2}\right],
\end{aligned}
\]
by It\^{o}'s isometry, and
\begin{equation}\label{eq Vpart}
0=\lim_{n_{2}\uparrow\infty}\lim_{n_{1}\uparrow\infty}\mathbb{E}\left[\left(\int_{0}^{\infty}\left|\Theta_{i}^{n_{1},n_{2}}(t)\right|\mathrm{d}|V_{i}(t)|\right)^{2}\right].
\end{equation}

Since $\int_{0}^{\cdot}\Theta_{i}^{n_{1},n_{2}}(t)\mathrm{d}M_{i}(t)$ is a uniformly integrable martingale (as it is a local martingale with bounded quadratic variation), Doob's submartingale inequality yields
\[
\mathbb{E}\left[\left(\sup_{t\geq0}\left|\int_{0}^{t}\Theta_{i}^{n_{1},n_{2}}(u)\mathrm{d}M_{i}(u)\right|\right)^{2}\right]
\leq4\mathbb{E}\left[\left(\int_{0}^{\infty}\Theta_{i}^{n_{1},n_{2}}(t)\mathrm{d}M_{i}(t)\right)^{2}\right],
\]
which implies
\begin{equation}\label{eq Mpart}
0=\lim_{n_{2}\uparrow\infty}\lim_{n_{1}\uparrow\infty}\mathbb{E}\left[\left(\sup_{t\geq0}
\left|\int_{0}^{t}\Theta_{i}^{n_{1},n_{2}}(u)\mathrm{d}M_{i}(u)\right|\right)^{2}\right].
\end{equation}
Therefore, we have
\[
\begin{aligned}
&\lim_{n_{2}\uparrow\infty}\lim_{n_{1}\uparrow\infty}\mathbb{E}\left[\left(\sup_{t\geq0}\left|\int_{0}^{t}\Theta_{i}^{n_{1},n_{2}}(u)\mathrm{d}\mu_{i}(u)\right|\right)^{2}\right]\\
&\leq\lim_{n_{2}\uparrow\infty}\lim_{n_{1}\uparrow\infty}2\mathbb{E}\Bigg[\left(\sup_{t\geq0}\left|\int_{0}^{t}\Theta_{i}^{n_{1},n_{2}}(u)\mathrm{d}M_{i}(u)\right|\right)^{2}
+\left(\sup_{t\geq0}\left|\int_{0}^{t}\Theta_{i}^{n_{1},n_{2}}(u)\mathrm{d}V_{i}(u)\right|\right)^{2}\Bigg]\leq0,
\end{aligned}
\]
where the second inequality holds due to \eqref{eq Vpart} and \eqref{eq Mpart}.

Write
\[
E_{i}^{n_{1},n_{2}}=\mathbb{E}\left[\left(\sup_{t\geq0}\left|\int_{0}^{t}\Theta_{i}^{n_{1},n_{2}}(u)\mathrm{d}\mu_{i}(u)\right|\right)^{2}\right],\quad n_{1},n_{2}\in\mathbb{N},
\]
and
\[
E_{i}=\lim_{n_{2}\uparrow\infty}\lim_{n_{1}\uparrow\infty}E_{i}^{n_{1},n_{2}}.
\]
For each $n_{2}\in\mathbb{N}$, denote $E_{i}^{n_{2}}=\lim_{n_{1}\uparrow\infty}E_{i}^{n_{1},n_{2}}$. Then we can find a subsequence $\big(n_{1}(n_{2})\big)_{n_{2}\in\mathbb{N}}$ of $\mathbb{N}$ with $n_{1}(n_{2})\uparrow\infty$ as $n_{2}\uparrow\infty$ such that, for each $n_{2}\in\mathbb{N}$,
\[
\left|E_{i}^{n_{1}(n_{2}),n_{2}}-E_{i}^{n_{2}}\right|\leq\frac{1}{n_{2}}.
\]
Since
\[
\begin{aligned}
\left|E_{i}^{n_{1}(n_{2}),n_{2}}-E_{i}\right|&\leq\left|E_{i}^{n_{1}(n_{2}),n_{2}}-E_{i}^{n_{2}}\right|+\left|E_{i}^{n_{2}}-E_{i}\right|\\
&\leq\frac{1}{n_{2}}+\left|E_{i}^{n_{2}}-E_{i}\right|\rightarrow0\quad\text{as}~n_{2}\uparrow\infty,
\end{aligned}
\]
we have $\lim_{n_{2}\uparrow\infty}E_{i}^{n_{1}(n_{2}),n_{2}}=E_{i}=0$. This implies
\[
\lim_{n_{2}\uparrow\infty}\sup_{t\geq0}\left|\int_{0}^{t}\sum_{i=1}^{d}h_{i}^{n_{1}(n_{2}),n_{2}}\big(\Lambda(u),\mu(u)\big)\mathrm{d}\mu_{i}(u)
-\int_{0}^{t}\sum_{i=1}^{d}h_{i}\big(\Lambda(u),\mu(u)\big)\mathrm{d}\mu_{i}(u)\right|=0
\]
in $L^{2}$. Since convergence in $L^{2}$ implies almost sure convergence of a subsequence, we can find a random subsequence $\left(n_{2}^{k}\right)_{k\in\mathbb{N}}$ of $\mathbb{N}$ with $n_{2}^{k}\uparrow\infty$ as $k\uparrow\infty$ and write $n_{1}^{k}=n_{1}\left(n_{2}^{k}\right)$ such that
\[
\lim_{k\uparrow\infty}\sup_{t\geq0}\left|\int_{0}^{t}\sum_{i=1}^{d}h_{i}^{n_{1}^{k},n_{2}^{k}}\big(\Lambda(u),\mu(u)\big)\mathrm{d}\mu_{i}(u)
-\int_{0}^{t}\sum_{i=1}^{d}h_{i}\big(\Lambda(u),\mu(u)\big)\mathrm{d}\mu_{i}(u)\right|=0,
\]
a.s. This implies the assertion.
\end{proof}

\vspace{5mm}

\begin{lemma}\label{Lemma 3}
Fix $l\in\mathbb{N}$; let $\overline{\Lambda}(\cdot)$ be an $l$-dimensional continuous process of finite variation; let $\left(\overline{\Upsilon}_{u,n}(\cdot)\right)_{u\in\{1,\cdots,l\},n\in\mathbb{N}}$ be a family of processes with $\left(\overline{\Upsilon}_{u,n}(\cdot)\right)_{n\in\mathbb{N}}$ uniformly bounded, for each $u\in\{1,\cdots,l\}$; and let $\left(\overline{\Theta}_{n}(\cdot)\right)_{n\in\mathbb{N}}$ be a sequence of non-decreasing continuous processes. Define
\[
H_{n}(\cdot)=\int_{0}^{\cdot}\sum_{u=1}^{l}\overline{\Upsilon}_{u,n}(t)\mathrm{d}\overline{\Lambda}_{u}(t)+\overline{\Theta}_{n}(\cdot),\quad n\in\mathbb{N}.
\]
If $\lim_{n\uparrow\infty}H_{n}(\cdot)=H(\cdot)$, a.s., then $H(\cdot)$ is of finite variation.
\end{lemma}

\begin{proof}
The following steps are partially inspired by the proof of Lemma~3.3 in \cite{Jaber_2016}.

Since $\left(\overline{\Upsilon}_{1,n}(\cdot)\right)_{n\in\mathbb{N}}$ is uniformly bounded, the Koml\'{o}s theorem (see Theorem~1.3 in \cite{MR1712239}) yields the following. For each $n\in\mathbb{N}$, there exists a convex combination $\Upsilon^{1}_{1,n}(\cdot)\in\mathrm{Conv}\left(\overline{\Upsilon}_{1,k}(\cdot),k\geq n\right)$ such that $\left(\Upsilon^{1}_{1,n}(\cdot)\right)_{n\in\mathbb{N}}$ converges to some adapted bounded process $\Upsilon_{1}(\cdot)$. More precisely, for each $n\in\mathbb{N}$, we can find some random integer $N_{n}\geq0$ and $\left(w_{n}^{k}\right)_{n\leq k\leq N_{n}}\subset[0,1]$ such that
\[
\sum_{k=n}^{N_{n}}w_{n}^{k}=1\quad\text{and}\quad\Upsilon^{1}_{1,n}(\cdot)=\sum_{k=n}^{N_{n}}w_{n}^{k}\overline{\Upsilon}_{1,k}(\cdot).
\]
For each $n\in\mathbb{N}$, define
\[
H^{1}_{n}(\cdot)=\sum_{k=n}^{N_{n}}w_{n}^{k}H_{n}(\cdot),\quad\Theta^{1}_{n}(\cdot)=\sum_{k=n}^{N_{n}}w_{n}^{k}\overline{\Theta}_{k}(\cdot),
\quad\text{and}\quad\Upsilon^{1}_{u,n}(\cdot)=\sum_{k=n}^{N_{n}}w_{n}^{k}\overline{\Upsilon}_{u,k}(\cdot),
\]
for all $u\in\{2,\cdots,l\}$.

Since $\lim_{n\uparrow\infty}H_{n}(\cdot)=H(\cdot)$, a.s., we have
\[
\left|H_{n}^{1}(\cdot)-H(\cdot)\right|=\left|\sum_{k=n}^{N_{n}}w_{n}^{k}H_{k}(\cdot)-H(\cdot)\right|
\leq\sum_{k=n}^{N_{n}}w_{n}^{k}\left|H_{k}(\cdot)-H(\cdot)\right|\rightarrow0
\]
as $n\uparrow\infty$, which implies $\lim_{n\uparrow\infty}H_{n}^{1}(\cdot)=H(\cdot)$, a.s. Besides, $\Theta^{1}_{n}(\cdot)$ is non-decreasing, as it is a convex combination of non-decreasing processes.

Since $\left(\Upsilon^{1}_{2,n}(\cdot)\right)_{n\in\mathbb{N}}$ is also uniformly bounded, by the Koml\'{o}s theorem again, for each $n\in\mathbb{N}$, there exists another convex combination $\Upsilon^{2}_{2,n}(\cdot)\in\mathrm{Conv}\left(\Upsilon^{1}_{2,k}(\cdot),k\geq n\right)$ such that $\left(\Upsilon^{2}_{2,n}(\cdot)\right)_{n\in\mathbb{N}}$ converges to some adapted bounded process $\Upsilon_{2}(\cdot)$. With the same convex combination for each $n\in\mathbb{N}$, define $\Upsilon^{2}_{u,n}(\cdot)$, for all $u\in\{1,3,\cdots,l\}$, $H_{n}^{2}(\cdot)$, and similarly $\Theta_{n}^{2}(\cdot)$. In particular, $\left(\Upsilon^{2}_{1,n}(\cdot)\right)_{n\in\mathbb{N}}$ still converges to $\Upsilon_{1}(\cdot)$, as for each $n\in\mathbb{N}$, $\Upsilon^{2}_{1,n}(\cdot)$ is a convex combination of processes that converge to $\Upsilon_{1}(\cdot)$. Similarly, we have $\lim_{n\uparrow\infty}H_{n}^{2}(\cdot)=H(\cdot)$, a.s. Moreover, $\Theta^{2}_{n}(\cdot)$ is non-decreasing.

Iteratively, we construct sequences of processes $\left(\Upsilon^{3}_{u,n}(\cdot)\right)_{n\in\mathbb{N}},\cdots,\left(\Upsilon^{l}_{u,n}(\cdot)\right)_{n\in\mathbb{N}}$, for each $u\in\{1,\cdots,l\}$, and processes $H_{n}^{3}(\cdot),\cdots,H_{n}^{l}(\cdot)$ and $\Theta_{n}^{3}(\cdot),\cdots,\Theta_{n}^{l}(\cdot)$ in the same manner. In particular, $\left(\Upsilon^{l}_{u,n}(\cdot)\right)_{n\in\mathbb{N}}$ converges to some adapted bounded process $\Upsilon_{u}$, for each $u\in\{1,\cdots,l\}$, and we have $\lim_{n\uparrow\infty}H_{n}^{l}(\cdot)=H(\cdot)$, a.s. Moreover, $\Theta^{l}_{n}(\cdot)$ is non-decreasing.

By the dominated convergence theorem, we have
\[
\lim_{n\uparrow\infty}\int_{0}^{\cdot}\sum_{u=1}^{l}\Upsilon^{l}_{u,n}(t)\mathrm{d}\overline{\Lambda}_{u}(t)
=\int_{0}^{\cdot}\sum_{u=1}^{l}\Upsilon_{u}(t)\mathrm{d}\overline{\Lambda}_{u}(t),\quad\text{a.s.},
\]
which is of finite variation. Therefore, we have
\[
H(\cdot)=\lim_{n\uparrow\infty}H^{l}_{n}(\cdot)=\int_{0}^{\cdot}\sum_{u=1}^{l}\Upsilon_{u}(t)\mathrm{d}\overline{\Lambda}_{u}(t)+\lim_{n\uparrow\infty}\Theta^{l}_{n}(\cdot),
\quad\text{a.s.}
\]
Since $\Theta^{l}_{n}(\cdot)$ is non-decreasing and converges, it is of finite variation, which implies the assertion.
\end{proof}

\subsection{Proof of Theorem \ref{Theorem new1}}

\begin{proof}[Proof of Theorem \ref{Theorem new1}.]
Assume that the semimartingale $\mu(\cdot)$ has the Doob-Meyer decomposition $\mu(\cdot)=\mu(0)+M(\cdot)+V(\cdot)$, where $M(\cdot)$ is a $d$-dimensional continuous local martingale and $V(\cdot)$ is a $d$-dimensional finite variation process with $M(0)=V(0)=0$.

Let $(\mathcal{W}_{n})_{n\in\mathbb{N}}$ be a non-decreasing sequence of open sets such that the closure of $\mathcal{W}_{n}$ is in $\mathcal{W}$, for all $n\in\mathbb{N}$. For each $\kappa\in\mathbb{N}$, we consider the stopping time
\begin{equation}\label{eq tauk}
\begin{aligned}
\tau_{\kappa}=\inf\Bigg\{&t\geq0;~\big(\Lambda(t),\mu(t)\big)\notin\mathcal{W}_{\kappa}&\\
&\text{or}~\sum_{i,j=1}^{d}\left|\langle M_{i},M_{j}\rangle\right|(t)+\sum_{i=1}^{d}\int_{0}^{t}\mathrm{d}|V_{i}(u)|+\sum_{v=1}^{m}\int_{0}^{t}\mathrm{d}|\Lambda_{v}(u)|\geq \kappa\Bigg\}&
\end{aligned}
\end{equation}
with $\inf\{\emptyset\}=\infty$. Since $\big(\Lambda(\cdot),\mu(\cdot)\big)\in\mathcal{W}$, we have $\lim_{\kappa\uparrow\infty}\tau_{\kappa}=\infty$, a.s. As $\bigcup_{\kappa\in\mathbb{N}}\{\tau_{\kappa}>t\}=\Omega$, for all $t\geq0$, to prove that $G$ is regular (Lyapunov), it is equivalent to show that $G$ is regular (Lyapunov) for $\Lambda\left(\cdot\wedge\tau_{\kappa}\right)$ and $\mu\left(\cdot\wedge\tau_{\kappa}\right)$, for all $\kappa\in\mathbb{N}$. Hence, without loss of generality, let us assume that $\big(\Lambda(\cdot),\mu(\cdot)\big)=\big(\Lambda(\cdot\wedge\tau_{\kappa}),\mu(\cdot\wedge\tau_{\kappa})\big)$, for some $\kappa\in\mathbb{N}$.

Without loss of generality, assume that $a_{ij}(\cdot)$ is a predictable and uniformly bounded process, for all $i,j\in\{1,\cdots,d\}$, such that
\[
\langle \mu_{i},\mu_{j}\rangle(t)=\int_{0}^{t}a_{ij}(u)\mathrm{d}A(u)\leq\kappa,\quad t\geq0,
\]
where $A(\cdot)=\sum_{i=1}^{d}\langle\mu_{i},\mu_{i}\rangle(\cdot)$. Here, the equality holds according to the Kunita–Watanabe theorem and the inequality due to \eqref{eq tauk}.

Now, consider a mollification $\left(G_{n_{1},n_{2}}\right)_{n_{1},n_{2}\in\mathbb{N}}$ of $G$ defined as in \eqref{eq mollification 1}. By Lemma~\ref{Lemma 1}, It\^{o}'s lemma applied to $G_{n_{1},n_{2}}$ yields
\begin{equation}\label{eq lemma.3.1}
\begin{aligned}
G_{n_{1},n_{2}}\big(\Lambda(t),\mu(t)\big)&=G_{n_{1},n_{2}}\big(\Lambda(0),\mu(0)\big)+\int_{0}^{t}\sum_{i=1}^{d}\frac{\partial G_{n_{1},n_{2}}}{\partial x_{i}}\big(\Lambda(u),\mu(u)\big)\mathrm{d}\mu_{i}(u)&\\
&+\int_{0}^{t}\Upsilon_{0,n_{1},n_{2}}(u)\mathrm{d}A(u)+\int_{0}^{t}\sum_{v=1}^{m}\Upsilon_{v,n_{1},n_{2}}(u)\mathrm{d}\Lambda_{v}(u),&
\end{aligned}
\end{equation}
for all $t\geq0$, where
\[
\Upsilon_{0,n_{1},n_{2}}(t)=\frac{1}{2}\sum_{i,j=1}^{d}\frac{\partial^{2}G_{n_{1},n_{2}}}{\partial x_{i}\partial x_{j}}\big(\Lambda(t),\mu(t)\big)a_{ij}(t)\quad\text{and}\quad \Upsilon_{v,n_{1},n_{2}}(t)=\frac{\partial G_{n_{1},n_{2}}}{\partial \lambda_{v}}\big(\Lambda(t),\mu(t)\big),
\]
for all $v\in\{1,\cdots,m\}$.

For all $(\lambda,x)\in\mathcal{W}$ and $i\in\{1,\cdots,d\}$, if (bi) holds, Lemma~\ref{Lemma 1}(iii) yields
\[
\lim_{n_{2}\uparrow\infty}\lim_{n_{1}\uparrow\infty}\frac{\partial G_{n_{1},n_{2}}}{\partial x_{i}}(\lambda,x)=\frac{\partial G}{\partial x_{i}}(\lambda,x);
\]
if (bii) holds, Lemma~\ref{Lemma 4} yields
\[
\lim_{n_{2}\uparrow\infty}\lim_{n_{1}\uparrow\infty}\frac{\partial G_{n_{1},n_{2}}}{\partial x_{i}}(\lambda,x)=f_{i}(\lambda,x),
\]
for some measurable function $f_{i}$. Then according to Lemma~\ref{Lemma 2}, we can find random subsequences $\left(n_{1}^{k}\right)_{k\in\mathbb{N}}$ and $\left(n_{2}^{k}\right)_{k\in\mathbb{N}}$ with $\lim_{k\uparrow\infty}n_{1}^{k}=\infty=\lim_{k\uparrow\infty}n_{2}^{k}$ such that, if we write $G_{k}=G_{n_{1}^{k},n_{2}^{k}}$, we have
\begin{equation}\label{eq F}
\lim_{k\uparrow\infty}\int_{0}^{t}\sum_{i=1}^{d}\frac{\partial G_{k}}{\partial x_{i}}\big(\Lambda(u),\mu(u)\big)\mathrm{d}\mu_{i}(u)=F\big(\Lambda(t),\mu(t)\big),\quad\text{a.s.},
\end{equation}
for all $t\geq0$, where
\[
\begin{aligned}
F\big(\Lambda(t),\mu(t)\big)=
\begin{cases}
\int_{0}^{t}\sum_{i=1}^{d}\frac{\partial G}{\partial x_{i}}\big(\Lambda(u),\mu(u)\big)\mathrm{d}\mu_{i}(u),\quad&\text{if (bi) holds}\\
\int_{0}^{t}\sum_{i=1}^{d}f_{i}\big(\Lambda(u),\mu(u)\big)\mathrm{d}\mu_{i}(u),&\text{if (bii) holds}
\end{cases}.
\end{aligned}
\]

To proceed, write
\[
H_{k}(t)=G_{k}\big(\Lambda(0),\mu(0)\big)-G_{k}\big(\Lambda(t),\mu(t)\big)+\int_{0}^{t}\sum_{i=1}^{d}\frac{\partial G_{k}}{\partial x_{i}}\big(\Lambda(u),\mu(u)\big)\mathrm{d}\mu_{i}(u),
\]
for all $k\in\mathbb{N}$, and
\[
H(t)=G\big(\Lambda(0),\mu(0)\big)-G\big(\Lambda(t),\mu(t)\big)+F\big(\Lambda(t),\mu(t)\big),
\]
for all $t\geq0$. Then, \eqref{eq lemma.3.1} with respect to the random subsequences $\left(n_{1}^{k}\right)_{k\in\mathbb{N}}$ and $\left(n_{2}^{k}\right)_{k\in\mathbb{N}}$ is of the form
\[
H_{k}(t)=-\int_{0}^{t}\Upsilon_{0,k}(u)\mathrm{d}A(u)-\int_{0}^{t}\sum_{v=1}^{m}\Upsilon_{v,k}(u)\mathrm{d}\Lambda_{v}(u),\quad t\geq0.
\]
Note that by Lemma~\ref{Lemma 1}(i) and \eqref{eq F}, $\lim_{k\uparrow\infty}H_{k}(t)=H(t)$, a.s., for all $t\geq0$.

A measurable function $DG$ in Condition~$1$ of Definition~\ref{Definition 3.1} is chosen with components
\[
\begin{aligned}
D_{i}G\big(\lambda,x\big)=
\begin{cases}
\frac{\partial G}{\partial x_{i}}(\lambda,x),\quad&\text{if (bi) holds}\\
f_{i}(\lambda,x),&\text{if (bii) holds}
\end{cases},\quad i\in\{1,\cdots,d\}.
\end{aligned}
\]
Then, as $\Gamma^{G}(\cdot)=H(\cdot)$ according to \eqref{eq 3.2}, it is enough to show that $H(\cdot)$ is of finite variation in the following four cases.

\textit{Case 1.}

Assume that (ai) and (bi) hold. Then by Lemma~\ref{Lemma 1}, the processes $\big(\Upsilon_{0,k}(\cdot)\big)_{k\in\mathbb{N}}$ and $\big(\Upsilon_{v,k}(\cdot)\big)_{v\in\{1,\cdots,m\},k\in\mathbb{N}}$ are uniformly bounded. With $l=m+1$, $\overline{\Lambda}_{v}(\cdot)=\Lambda_{v}(\cdot)$ and $\big(\overline{\Upsilon}_{v,k}(\cdot)\big)_{k\in\mathbb{N}}=\big(\Upsilon_{v,k}(\cdot)\big)_{k\in\mathbb{N}}$, for all $v\in\{1,\cdots,m\}$, $\overline{\Lambda}_{m+1}(\cdot)=A(\cdot)$, $\big(\overline{\Upsilon}_{m+1,k}(\cdot)\big)_{k\in\mathbb{N}}=\big(\Upsilon_{0,k}(\cdot)\big)_{k\in\mathbb{N}}$, and $\left(\overline{\Theta}_{k}(\cdot)\right)_{k\in\mathbb{N}}=0$, Lemma~\ref{Lemma 3} yields that $H(\cdot)$ is of finite variation on compact sets.

\textit{Case 2.}

Assume that (ai) and (bii) hold. By Lemma~\ref{Lemma 1}(iii), the processes $\big(\Upsilon_{v,k}(\cdot)\big)_{v\in\{1,\cdots,m\},k\in\mathbb{N}}$ are uniformly bounded. Since $G$ is concave in the second argument, for each $k\in\mathbb{N}$, $G_{k}$ is also concave in the second argument. As a consequence, the matrix
\[
\nabla^{2}G_{k}=\left(\frac{\partial^{2}G_{k}}{\partial x_{i}\partial x_{j}}\right)_{i,j\in\{1,\cdots,d\}}
\]
is negative semidefinite. Note that the matrix-valued process $a(\cdot)=\big(a_{ij}(\cdot)\big)_{i,j\in\{1,\cdots,d\}}$ can be chosen to be symmetric and positive semidefinite. Hence, we can find a matrix-valued process $\sigma(\cdot)=\big(\sigma_{ij}(\cdot)\big)_{i,j\in\{1,\cdots,d\}}$ such that $a(\cdot)=\sigma(\cdot)\sigma'(\cdot)$, which yields $a_{ij}(\cdot)=\sum_{l=1}^{d}\sigma_{il}(\cdot)\sigma_{jl}(\cdot)$, for all $i,j\in\{1,\cdots,d\}$. In this case,
\[
\begin{aligned}
\sum_{i,j=1}^{d}\frac{\partial^{2}G_{k}}{\partial x_{i}\partial x_{j}}\big(\Lambda(t),\mu(t)\big)a_{ij}(t)&=\sum_{i,j=1}^{d}\frac{\partial^{2}G_{k}}{\partial x_{i}\partial x_{j}}\big(\Lambda(t),\mu(t)\big)\sum_{l=1}^{d}\sigma_{il}(t)\sigma_{jl}(t)\\
&=\sum_{l=1}^{d}\sigma_{l}(t)\nabla^{2}G_{k}\big(\Lambda(t),\mu(t)\big)\sigma_{l}'(t)\leq0,
\end{aligned}
\]
for all $t\geq0$, where $\sigma_{l}(\cdot)$ is the $l$-th row of $\sigma(\cdot)$. Hence, $\Upsilon_{0,k}(t)\leq0$, for all $t\geq0$, which implies that the processes
\[
\overline{\Theta}_{k}(\cdot)=-\int_{0}^{\cdot}\Upsilon_{0,k}(t)\mathrm{d}A(t),\quad k\in\mathbb{N},
\]
are non-decreasing. Similar to Case~1, but now with $l=m$, Lemma~\ref{Lemma 3} yields again that $H(\cdot)$ is of finite variation.

\textit{Case 3.}

Assume that (aii) and (bi) hold. By Lemma~\ref{Lemma 1}(v), the process $\big(\Upsilon_{0,k}(\cdot)\big)_{k\in\mathbb{N}}$ is uniformly bounded. As $G$ is non-increasing in the $v$-th dimension of the first argument, so is $G_{k}$, for all $v\in\{1,\cdots,m\}$. Therefore, $\Upsilon_{v,k}(t)\leq0$, for all $t\geq0$, as $\Lambda(\cdot)$ is non-decreasing in the $v$-th dimension, for all $v\in\{1,\cdots,m\}$. This implies that the processes
\[
\overline{\Theta}_{k}(\cdot)=-\int_{0}^{\cdot}\sum_{v=1}^{m}\Upsilon_{v,k}(t)\mathrm{d}\Lambda_{v}(t),\quad k\in\mathbb{N},
\]
are non-decreasing. Similar to above, Lemma~\ref{Lemma 3} implies that $H(\cdot)$ is of finite variation.

\textit{Case 4.}

Assume that (aii) and (bii) hold. With
\[
\overline{\Theta}_{k}(\cdot)=-\int_{0}^{\cdot}\Upsilon_{0,k}(t)\mathrm{d}A(t)-\int_{0}^{\cdot}\sum_{v=1}^{m}\Upsilon_{v,k}(t)\mathrm{d}\Lambda_{v}(t),\quad k\in\mathbb{N},
\]
Lemma~\ref{Lemma 3} implies again that $H(\cdot)$ is of finite variation. It is clear that $G$ is Lyapunov.
\end{proof}

\subsection{Proof of Theorem \ref{Theorem new2}}

\begin{proof}[Proof of Theorem~\ref{Theorem new2}.]
The following steps are partially inspired by the proof of Theorem~3.8 in \cite{MR3663643}. According to Theorem~2.3 in \cite{MR2428716}, for each $l\in\{1,\cdots,d\}$, one can find a measurable function $\boldsymbol{h}_{l}:\Delta^{d}\rightarrow(0,1]$ and a finite variation process $\boldsymbol{B}_{l}(\cdot)$ with $\boldsymbol{B}_{l}(0)=0$ such that
\begin{equation}\label{eq dmu1}
\boldsymbol{\mu}_{l}(\cdot)=\boldsymbol{\mu}_{l}(0)+\int_{0}^{\cdot}\sum_{i=1}^{d}\boldsymbol{h}_{l}\big(\mu(t)\big)
\boldsymbol{1}_{\left\{\mu_{(l)}(t)=\mu_{i}(t)\right\}}\mathrm{d}\mu_{i}(t)+\boldsymbol{B}_{l}(\cdot).
\end{equation}

Since $\boldsymbol{G}$ is regular for $\Lambda(\cdot)$ and $\boldsymbol{\mu}(\cdot)$, by Definition~\ref{Definition 3.1}, there exist a measurable function $D\boldsymbol{G}$ and a finite variation process $\Gamma^{\boldsymbol{G}}(\cdot)$ such that
\begin{equation}\label{eq G1}
\boldsymbol{G}\big(\Lambda(\cdot),\boldsymbol{\mu}(\cdot)\big)=\boldsymbol{G}\big(\Lambda(0),\boldsymbol{\mu}(0)\big)
+\int_{0}^{\cdot}\sum_{l=1}^{d}D_{l}\boldsymbol{G}\big(\Lambda(t),\boldsymbol{\mu}(t)\big)\mathrm{d}\boldsymbol{\mu}_{l}(t)-\Gamma^{\boldsymbol{G}}(\cdot).
\end{equation}
By \eqref{eq dmu1}, we have
\begin{equation}\label{eq Gamma1}
\begin{aligned}
\int_{0}^{\cdot}\sum_{l=1}^{d}D_{l}\boldsymbol{G}\big(\Lambda(t),\boldsymbol{\mu}(t)\big)\mathrm{d}\boldsymbol{\mu}_{l}(t)=&
\int_{0}^{\cdot}\sum_{l=1}^{d}D_{l}\boldsymbol{G}\big(\Lambda(t),\boldsymbol{\mu}(t)\big)\boldsymbol{h}_{l}\big(\mu(t)\big)
\boldsymbol{1}_{\left\{\mu_{(l)}(t)=\mu_{i}(t)\right\}}\mathrm{d}\mu_{i}(t)\\
&+\int_{0}^{\cdot}\sum_{l=1}^{d}D_{l}\boldsymbol{G}\big(\Lambda(t),\boldsymbol{\mu}(t)\big)\mathrm{d}\boldsymbol{B}_{l}(t).
\end{aligned}
\end{equation}

Now consider the measurable function $DG:\mathcal{W}\rightarrow\mathbb{R}^{d}$ with components
\[
D_{i}G(\lambda,x)=\sum_{l=1}^{d}D_{l}\boldsymbol{G}\big(\lambda,\mathfrak{R}(x)\big)\boldsymbol{h}_{l}(x)\boldsymbol{1}_{\left\{x_{(l)}=x_{i}\right\}},\quad i\in\{1,\cdots,d\},
\]
and the finite variation process
\[
\Gamma^{G}(\cdot)=\Gamma^{\boldsymbol{G}}(\cdot)-\int_{0}^{\cdot}\sum_{l=1}^{d}D_{l}\boldsymbol{G}\big(\Lambda(t),\boldsymbol{\mu}(t)\big)\mathrm{d}\boldsymbol{B}_{l}(t).
\]
Then \eqref{eq G1} and \eqref{eq Gamma1}, together with $G(\lambda,x)=\boldsymbol{G}\big(\lambda,\mathfrak{R}(x)\big)$, yield \eqref{eq 3.2}, i.e., $G$ is regular for $\Lambda(\cdot)$ and $\mu(\cdot)$.
\end{proof}

\subsection{An alternative proof for a special case}

The proof technique of Theorem~VII.31 in \cite{MR745449} suggests an alternative argument for the case that conditions (ai) and (bii) in Theorem~\ref{Theorem new1} hold. We summarize these ideas in the following result.
\\
\begin{theorem}\label{Theorem Dellacherie}
If a function $f:\mathcal{W}\rightarrow\mathbb{R}$ is locally Lipschitz in the first argument and concave in the second argument, then the process $f\big(\Lambda(\cdot),\mu(\cdot)\big)$ is a semimartingale.
\end{theorem}

\begin{proof}
Assume that the semimartingale $\mu(\cdot)$ has the Doob-Meyer decomposition $\mu(\cdot)=\mu(0)+M(\cdot)+V(\cdot)$, where $M(\cdot)$ is a $d$-dimensional continuous local martingale and $V(\cdot)$ is a $d$-dimensional finite variation process with $M(0)=V(0)=0$.

Let $(\mathcal{W}_{n})_{n\in\mathbb{N}}$ be a non-decreasing sequence of open sets such that the closure of $\mathcal{W}_{n}$ is in $\mathcal{W}$, for all $n\in\mathbb{N}$. For each $\kappa\in\mathbb{N}$, we consider the stopping time $\tau_{\kappa}$ given in \eqref{eq tauk}. Without loss of generality, let us assume again that $\big(\Lambda(\cdot),\mu(\cdot)\big)=\big(\Lambda(\cdot\wedge\tau_{\kappa}),\mu(\cdot\wedge\tau_{\kappa})\big)$, for some $\kappa\in\mathbb{N}$.

Since $f$ is locally Lipschitz in both arguments (see Theorem~10.4 in \cite{MR0274683}), we can find a Lipschitz constant $L$ such that, for all $s,t\geq0$ with $s\leq t$, we have
\begin{equation}\label{eq L}
\begin{aligned}
\big|f&\big(\Lambda(t),\mu(t)\big)-f\big(\Lambda(s),\mu(0)+M(t)+V(s)\big)\big|\\
&\leq L\left(\sum_{v=1}^{m}\big|\Lambda_{v}(t)-\Lambda_{v}(s)\big|+\sum_{i=1}^{d}\big|V_{i}(t)-V_{i}(s)\big|\right)\\
&\leq L\left(\sum_{v=1}^{m}\int_{s}^{t}\big|\mathrm{d}\Lambda_{v}(u)\big|
+\sum_{i=1}^{d}\int_{s}^{t}\big|\mathrm{d}V_{i}(u)\big|\right).
\end{aligned}
\end{equation}

Let
\[
Z(\cdot)=-f\big(\Lambda(\cdot),\mu(\cdot)\big)+L\left(\sum_{v=1}^{m}\int_{0}^{\cdot}\big|\mathrm{d}\Lambda_{v}(t)\big|
+\sum_{i=1}^{d}\int_{0}^{\cdot}\big|\mathrm{d}V_{i}(t)\big|\right),
\]
then $Z(\cdot)$ is bounded. Hence we have
\[
\begin{aligned}
\mathbb{E}\left[Z(t)-Z(s)|\mathcal{F}(s)\right]&=\mathbb{E}\left[f\big(\Lambda(s),\mu(s)\big)-f\big(\Lambda(s),\mu(0)+M(t)+V(s)\big)|\mathcal{F}(s)\right]\\
&+\mathbb{E}\Bigg[f\big(\Lambda(s),\mu(0)+M(t)+V(s)\big)-f\big(\Lambda(t),\mu(t)\big)\\
&+L\left(\sum_{v=1}^{m}\int_{s}^{t}\big|\mathrm{d}\Lambda_{v}(u)\big|
+\sum_{i=1}^{d}\int_{s}^{t}\big|\mathrm{d}V_{i}(u)\big|\right)\Big|\mathcal{F}(s)\Bigg]\\
&\geq\mathbb{E}\left[f\big(\Lambda(s),\mu(s)\big)-f\big(\Lambda(s),\mu(0)+M(t)+V(s)\big)|\mathcal{F}(s)\right]\geq0,
\end{aligned}
\]
where the first inequality is by \eqref{eq L} and the second inequality holds by Jensen's inequality. Therefore, $Z(\cdot)$ is a submartingale, which makes $f\big(\Lambda(\cdot),\mu(\cdot)\big)$ a semimartingale.
\end{proof}

\vspace{0.5cm}

\bibliographystyle{apalike}
\bibliography{Lyapunov}

\end{document}